\documentclass[journal]{IEEEtran}
\usepackage{cite}
\usepackage{amsmath,amssymb,amsthm,multirow,dsfont,cuted}
\usepackage{enumitem}
\usepackage{epsfig}
\usepackage{tikz,pgfplots}
\usepackage{algorithm}
\usepackage{algpseudocode}
\usepackage{diagbox}
\usepackage{extarrows}
\usepackage{graphics,graphicx}
\usepackage{subfigure}

\newtheorem{theorem}{Theorem}

\newtheorem{corollary}{Corollary}
\newtheorem{assumption}{Assumption}

\newcommand{\cE}{\mathcal{E}}

\newcommand{\cO}{\mathcal{O}}
\newcommand{\N}{\mathcal{N}}

\newcommand{\C}{\mathcal{C}}

\newcommand{\bA}{\boldsymbol{A}}

\newcommand{\bD}{\boldsymbol{D}}
\newcommand{\bE}{\boldsymbol{E}}

\newcommand{\bH}{\boldsymbol{H}}
\newcommand{\bI}{\boldsymbol{I}}

\newcommand{\bL}{\boldsymbol{L}}

\newcommand{\bP}{\boldsymbol{P}}

\newcommand{\bR}{\boldsymbol{R}}
\newcommand{\bS}{\boldsymbol{S}}

\newcommand{\bU}{\boldsymbol{U}}
\newcommand{\bV}{\boldsymbol{V}}
\newcommand{\bW}{\boldsymbol{W}}
\newcommand{\bX}{\boldsymbol{X}}
\newcommand{\bY}{\boldsymbol{Y}}
\newcommand{\bZ}{\boldsymbol{Z}}

\newcommand{\boldh}{\boldsymbol{h}}
\newcommand{\bp}{\boldsymbol{p}}
\newcommand{\br}{\boldsymbol{r}}

\newcommand{\bv}{\boldsymbol{v}}

\newcommand{\bw}{\boldsymbol{w}}
\newcommand{\bx}{\boldsymbol{x}}
\newcommand{\by}{\boldsymbol{y}}
\newcommand{\bz}{\boldsymbol{z}}

\newcommand{\bpsi}{\boldsymbol{\psi}}

\newcommand{\bSig}{\boldsymbol{\Sigma}}
\newcommand{\bsig}{\boldsymbol{\sigma}}

\def\One{\boldsymbol{1}}
\def\Zero{\boldsymbol{0}}

\newcommand{\cA}{\boldsymbol{\cal{A}}}
\newcommand{\cB}{\boldsymbol{\cal{B}}}

\newcommand{\cD}{\boldsymbol{\cal{D}}}
\newcommand{\cF}{\boldsymbol{\cal{F}}}
\newcommand{\cG}{\boldsymbol{\cal{G}}}

\newcommand{\cM}{\boldsymbol{\cal{M}}}

\newcommand{\cR}{\boldsymbol{\cal{R}}}
\newcommand{\cS}{\boldsymbol{\cal{S}}}

\newcommand{\tr}{\text{Tr}}
\newcommand{\vc}{\text{vec}}

\newcommand{\col}{\text{col}}

\newcommand{\diag}{\text{diag}}
\newcommand{\bdiag}{\text{bdiag}}

\newcommand{\expec}{\mathbb{E}}

\DeclareMathOperator*{\minimize}{minimize}

\begin{document}
	
\title{Online Distributed Learning over Graphs with Multitask Graph-Filter Models}
	
\author{Fei~Hua, 
		Roula~Nassif,~\IEEEmembership{Member,~IEEE,}
		C{\'e}dric~Richard,~\IEEEmembership{Senior~Member,~IEEE,}
		Haiyan~Wang
		and~Ali~H.~Sayed,~\IEEEmembership{Fellow,~IEEE}% <-this % stops a space
\thanks{The work of C. Richard was funded in part by the PIA program under its IDEX UCAJEDI project (ANR-15-IDEX-0001) and by ANR under grant ANR-19-CE48-0002. The work of F. Hua was supported in part by China Scholarship Council and NSFC grant 61771396. The work of H. Wang was partially supported by NSFC under grants 61571365, 61671386 and National Key R\&D Program of China (2016YFC1400200). The work of A.~H.~Sayed was supported  in part by NSF grant CCF-1524250. Part of the results in this paper were presented at EUSIPCO~\cite{hua2018preconditioned} and ASILOMAR~\cite{hua2018decentralized} conferences.}% <-this % stops a space
\thanks{F. Hua  and H. Wang are with School of Marine Science and Technology, Northwestern Polytechnical University, Xi’an 710072, China. F. Hua is also with the Universit\'{e} C\^{o}te d'Azur, OCA, CNRS, 06108 Nice, France. H. Wang is also with the School of Electronic Information and Artificial Intelligence, Shaanxi University of Science and Technology, Xi’an 710021, China (emails: fei.hua@oca.eu; hywang@sust.edu.cn).}% <-this % stops a space
\thanks{C.~Richard is  with the Universit\'{e} C\^{o}te d'Azur, OCA, CNRS,  06108 Nice, France (email: cedric.richard@unice.fr).}% <-this % stops a space
\thanks{R.~Nassif and A.~H.~Sayed are  with the School of Engineering, Ecole Polytechnique F\'{e}d\'{e}rale de Lausanne, CH-1015 Lausanne, Switzerland (emails: roula.nassif@epfl.ch;   ali.sayed@epfl.ch).}% <-this % stops a space
}

\markboth{IEEE TRANSACTIONS ON SIGNAL AND INFORMATION PROCESSING OVER NETWORKS,~Vol.~xx, No.~x, xx~2020}%
{Hua \MakeLowercase{\textit{et al.}}: Online Distributed Learning over Graphs with Multitask Graph-Filter Models}
	
	% \IEEEpubid{0000--0000/00\$00.00~\copyright~2019 IEEE}
	% column for its text to clear the IEEEpubid mark.m
	
\maketitle
\begin{abstract}	
		In this work, we are interested in adaptive and distributed estimation of graph filters from streaming data. We formulate this problem as a consensus estimation problem over graphs, which can be addressed with diffusion LMS strategies. Most popular graph-shift operators such as those based on the graph Laplacian matrix, or the adjacency matrix, are not energy preserving. This may result in an ill-conditioned estimation problem, and reduce the convergence speed of the distributed algorithms. To address this issue and improve the transient performance, we introduce a preconditioned graph diffusion LMS algorithm. We also propose a computationally efficient version of this algorithm by approximating the Hessian matrix with local information. Performance analyses in the mean and mean-square sense are provided. Finally, we consider a more general problem where the filter coefficients to estimate may vary over the graph. To avoid a large estimation bias, we introduce an unsupervised clustering method for splitting the global estimation problem into local ones. Numerical results show the effectiveness of the proposed algorithms and validate the theoretical results. 
\end{abstract}
	
	% \begin{center} \boldfseries EDICS Category: 3-BBND \end{center}
	% \fi
	% creates the second title. It will be ignored for other modes.m
	\IEEEpeerreviewmaketitle
	
	%%%%%
\section{Introduction}
\IEEEPARstart{D}{ata} generated by network-structured applications often exhibit non-Euclidean structures, which make traditional signal processing techniques inefficient to analyze them. In contrast, graph signal processing (GSP) provides useful tools to analyze and process signals on graphs. They represent them as samples at the vertices of a possibly weighted graph, and use algebraic and spectral properties of the graph to study the signals. These graph representations are useful in applications ranging from social and economic networks to smart grids~\cite{djuric2018cooperative,ortega2017graph,sandryhaila2014big,shuman2013emerging}. Recent results in the area include sampling~\cite{chen2015discrete,anis2016efficient,Tsitsvero2016Signals}, filtering~\cite{sandryhaila2013discrete}, and inference and learning~\cite{nassif2019Regularization,defferrard2016convolutional,gama2018convolutional,anis2019sampling}, to cite a few. 
	
In order to cope with graph signals, GSP relies on two  ingredients: the graph shift operator (GSO) on  one hand, which accounts for the topology of the graph, and the graph Fourier transform (GFT) on the other hand, which allows to represent graph signals in the graph frequency domain. Built upon the definition of the GFT, graph filters play a central role in processing graph signal spectra by selectively amplifying or attenuating frequency components. Various architectures of graph filters have been proposed in the literature, including finite impulse response (FIR)~\cite{shuman2013emerging,sandryhaila2013discrete} and infinite impulse response (IIR)~\cite{shi2015infinite,liu2019filter} filters. From a perspective of scalability, and considering energy constraints and band-limited communications that may be encountered in large networks of distributed nodes such as sensor networks, significant efforts have been made recently to derive distributed graph filters. These filtering procedures allow each node to exchange only  local information with its neighboring nodes~\cite{loukas2015distributed,isufi2017autoregressive,segarra2017optimal,coutino2018advances,shuman2018distributed}.
	
Much of the GSP literature has focused on static graph signals, that is, signals that need not evolve with time. However, a wide spectrum of network-structured problems requires adaptation to time-varying dynamics. Sensor networks, social networks, vehicular networks, communication networks, and power grids are some typical examples. Prior to the more recent  GSP literature, many earlier works on adaptive networks have addressed problems dealing with this challenge by developing processing strategies that are well-suited to data streaming into graphs; see, e.g.,~\cite{sayed2013diffusionb,sayed2014adaptive,sayed2013diffusion}. Several diffusion strategies have been introduced, and their performance studied in various situations, such as  diffusion LMS \cite{Lopes2008LMS}, RLS \cite{Cattivelli2008}, and APA~\cite{Li2009}. By referring to the problem of estimating  an optimal parameter vector at a node as a ``task'', and depending on  the relations between the  parameter vectors across the entire network, adaptive networks can be divided into single or multitask networks. In single-task networks, all nodes estimate the same parameter vector. Typical works include~\cite{sayed2013diffusionb,sayed2014adaptive,sayed2013diffusion}. With multitask networks, multiple but related parameter vectors are inferred simultaneously in a cooperative manner, so as to improve the estimation accuracy by using the similarities between tasks~\cite{chen2015diffusion,chen2014diffusion,nassif2014multitask,nassif2016proximal}.
	
\textit{Related works:} In this work, we are interested in online learning of linear graph models for representing streaming graph signals in a distributed manner. We focus on diffusion strategies because they are scalable, robust, and enable network adaptation and learning.
Recently, some research works have considered  time-varying graph signals. An adaptive graph signal reconstruction algorithm based on the LMS is proposed in~\cite{di2016adaptive} but it operates in a centralized manner.  In~\cite{mei2017signal}, the authors focus on estimating a network structure capturing the dependencies among time series in the form of a possibly directed, weighted adjacency matrix. A causal autoregressive process is introduced in the time series to capture the intuition that information travels over the network at some fixed speed. In~\cite{isufi2019forecasting}, vector autoregressive (VAR) and vector autoregressive moving average (VARMA) models are proposed  for predicting time-varying processes on graphs. Joint time-vertex stationarity is introduced for time-varying graph signals in~\cite{loukas2016stationary,loukas2016predicting}, and a joint time-vertex harmonic analysis for graph signals is proposed in~\cite{grassi2017time}. It is shown that joint stationarity facilitates estimation or recovery tasks when compared to purely time or graph methods.

All the aforementioned works focus on centralized solutions whereas distributed algorithms may be more appropriate within the context of big data applications.  In~\cite{isufi20162d,isufi2017filtering}, the behavior of some distributed graph filters on time-varying graph signals is studied. Considering graph signal sampling and reconstruction, several distributed algorithms have been proposed to track time-varying band-limited graph  signals,  e.g., LMS-based algorithms in~\cite{di2017distributed}, RLS-based methods in~\cite{DiLorenzo2018adaptive}, Kalman-based methods in~\cite{isufi2017observing}, and kernel-based algorithms in~\cite{romero2017kernel}. In~\cite{qiu2017time}, the authors are interested in time-varying graph signals with temporal smoothness prior. They devise  distributed gradient descent algorithms to reconstruct the signals. Most of these works assume that the graph signals are band-limited. Another limitation is that the graph Fourier decomposition (eigenvectors) is needed beforehand, which is impractical for large networks.

%\textit{Contributions:} \cred{The naturally dynamic of graph signals in many GSP applications requires proper formulation and methods to adapt and learn the time-varying graph signals or tasks. The GSP framework is intrinsic network-structured, this motivate us to blend the concepts from adaptive networks~\cite{sayed2014adaptation} and techniques in  graph signal processing. It is worthy noting that, several recent works have successfully applied adaptive algorithms on graph signal~\cite{di2016adaptive,di2017distributed,DiLorenzo2018adaptive} where the graph Fourier coefficients are learned from streaming graph signals under  band-limited assumption.  It allows adaptive reconstruction and tracking time-varying graph signals from partial observations. }

\textit{Contributions:} Recently, several works successfully applied adaptive algorithms to graph signals. In~\cite{di2016adaptive,di2017distributed,DiLorenzo2018adaptive} for instance, graph Fourier coefficients are learned from streaming graph signals under band-limited assumption to perform adaptive reconstruction and tracking of time-varying graph signals from partial observations. In this work, we are interested in online distributed learning of linear graph models without assumption of band-limited processes. We use graph filter models in the time-vertex domain where there is no need to decompose the graph shift operator.  The formulated optimization problem relies on minimizing a global cost consisting of the aggregate sum of individual costs at each vertex.  To address this problem, we blend concepts from adaptive networks~\cite{sayed2014adaptation} and graph signal processing to devise graph diffusion LMS strategies. Considering that most popular shift operators are not energy preserving and may result in a slow convergence speed, we introduce a preconditioned optimization strategy to improve the transient performance. As this may lead to an increased computational complexity, we further propose a computationally efficient algorithm. Explicit theoretical  performance analyses in the mean and mean-square-error sense are provided.  We also give alternative theoretical results that are tractable for large networks. Simulation results show the efficiency of the proposed algorithms and validate the theoretical models. Finally, we extend these node-invariant filter models to more flexible ones where each node in the graph seeks to estimate a local node-varying graph filter. 
This allows us to exploit more degrees of freedom in the filter coefficients to better model graph signals. We introduce an unsupervised clustering strategy to determine which nodes in the graph share the same graph filter and may collaborate to estimate its parameters. Numerical results on a real-word dataset illustrate the efficiency of the proposed methods.
	
The rest of the paper is organized  as follows. Section~\ref{sec:central} formulates the problem and provides the centralized solution. Section~\ref{sec:diffusion} introduces the distributed algorithms, namely, the graph diffusion LMS strategy and its preconditioned counterparts. Section~\ref{sec:analy} provides their theoretical analyses in the mean and the mean-square sense. A clustering method is devised to estimate local node-varying graph filters in Section~\ref{sec:multi}. Numerical results in Section~\ref{sec:simu} show the effectiveness of these algorithms and validate the theoretical models.
	
\medskip
	
\noindent\textbf{Notations:} We use normal font letters to denote scalars, boldface lowercase letters to denote column vectors, and boldface uppercase letters to denote matrices.  The  $m$-th entry of a vector $\bx$ is denoted by  $x_m$ or $[\bx]_m$ when needed, the~$(m,n)$-th entry of a matrix $\bX$ is denoted by $x_{mn}$ or $[\bX]_{m,n}$ when needed, the $k$-th row of a matrix $\bX$ is denoted by $[\bX]_{k,\bullet}$.  We use the symbol $\otimes$ to denote the Kronecker product and the symbol $\tr(\cdot)$ to denote the trace operator. The operator $\col\{\cdot\}$ stacks the column vector entries on top of each other. The symbol $\vc(\bX)$ refers to the  vectorization operator that stacks the columns of a matrix on top of each other. Operator $\bx = \diag(\bX)$ stores the diagonal entries of $\bX$ into vector $\bx$, and $\bX = \diag(\bx)$ is a diagonal matrix containing the vector $\bx$ along its main diagonal. The symbol $\bdiag\{\cdot\}$ forms a matrix from block arguments by placing each block immediately below and to the right of its predecessor. The symbol $\|\cdot \|_{b,\infty}$ denotes the block maximum norm of a matrix. The symbols $\One$ and $\bI$ are the vector of all ones and the identify matrix of appropriate size, respectively. The symbols $\rho(\cdot)$ and $\lambda_{\max}(\cdot)$ denote the spectral radius and the maximum eigenvalue of its matrix argument, respectively.
	
\section{Problem Formulation and Centralized Solution}\label{sec:central}
	
We consider an undirected, weighted and connected graph $\mathcal{G}=(\N, \cE, \bW)$ of $N$ nodes, where $\N=\{1, 2, \ldots, N\}$ is the set of nodes, and $\cE$ is the set of edges such that if node $k$ is connected to node $\ell$, then $(k,\ell)\in \cE$. We denote by $\N_k$  the neighborhood of node~$k$ including itself, namely, $\N_k=\{\ell: \ell=k \vee  (\ell,k)\in \cE  \}$. Matrix $\bW \in  \mathbb{R}^{N\times N}$ is the adjacency matrix  whose $(k,\ell)$-th entry $w_{k\ell}$  assigns  a weight to the relation between vertices $k$ and $\ell$. Since the graph is undirected,  $\bW$ is a symmetric matrix. The degree matrix $\bD \triangleq \diag( \bW \One)$  is a diagonal matrix whose $i$-th diagonal entry is the degree of node $i$ which is equal to the sum of all the weights of edges incident at node $i$. The combinatorial Laplacian matrix is defined as $\bL \triangleq \bD-\bW$ which is a real, symmetric, positive semi-definite matrix.   We further assume that the graph is endowed with a graph shift operator defined as an $N \times N$ shift matrix $\bS$ whose entry $s_{k\ell}$ can be non-zero only if $k=\ell$ or $(k,\ell)\in \cE$. Although the shift matrix can be any matrix that captures the graph topology for the problem at hand~\cite{ortega2017graph},  popular choices are  the  graph Laplacian matrix~\cite{shuman2013emerging}, the adjacency matrix~\cite{sandryhaila2013discrete}, and their normalized counterparts.  A graph signal  is  defined  as  $\bx=[x_1,\ldots,x_N ]^{\top} \in \mathbb{R}^N$ where $x_k$ is the signal sample associated with node $k$. Let $\bx(i)$ denote the graph signal at time $i$. Operation $\bS\bx$ is called graph  shift. It can be performed locally   at each node $k$ by linearly combining the samples $x_{\ell}$ from its neighboring nodes $\ell\in\N_k$.

\subsection{Graph filter and data model}
	
In this paper, we focus on linear shift-invariant FIR graph filters $\bH: {\mathbb{R}^{N\times N} \rightarrow \mathbb{R}^{N\times N}}$ {of order $M$}, which are polynomials of the graph-shift operator~\cite{sandryhaila2013discrete}: 
	\begin{equation}
	\bH \triangleq \sum_{m=0}^{M-1} h_{m}^o \bS^m, \label{eq:gf}
	\end{equation}
where $\boldh^o=\{h_{m}^o\}_{m=0}^{M-1}$ are the scalar filter  coefficients. With the definitions of graph signal and graph shift operator, one common filtering model assumes that the filtered graph signal $\by(i)$ is generated from the input graph signal $\bx(i)$ as follows~\cite{sandryhaila2013discrete,nassif2017graph}:  
	\begin{equation}
	\by(i)=\bH \bx(i)+\bv(i)= \sum_{m=0}^{M-1} h_{m}^o \bS^m\bx(i)+\bv(i), \label{eq:gfing}
	\end{equation}
where $\bv(i)=[v_1(i),\ldots,v_N(i) ]^{\top} \in \mathbb{R}^N$ denotes an i.i.d. zero-mean noise independent of any other signal and with covariance matrix $\bR_v=\diag\{\sigma_{v,k}^2 \}_{k=1}^{N}$. For each node $k$, the  filtered signal $y_k(i)$ can be computed by  linearly combining the input signals at nodes located in an $(M-1)$-hop neighborhood~\cite{shuman2013emerging}. This model however  assumes the instantaneous diffusion of information over the graph since $\bS^m \bx(i)$ translates $ \bx(i)$ without time delay. As this assumption may appear as a serious limitation, we consider the more general model embedding the temporal dimension as follows~\cite{isufi20162d,nassif2017distributed}:
	\begin{equation}
	\by(i)=\sum_{m=0}^{M-1} h_{m}^o \bS^m\bx(i-m)+\bv(i). \label{eq:gfingt}
	\end{equation}
Observe that the input signal $\bx(i)$ in~\eqref{eq:gfing} has been replaced by $\bx(i-m)$ in~\eqref{eq:gfingt}, i.e.,  the $m$-hop spatial shift $\bS^m$  is now carried out in $m$ time slots. This model implements an FIR filter in both  graph domain and temporal domain. By retaining the following shifted signals that form the $N\times M-1$ matrix:
	\begin{equation}
	\bX_{r}=\Big[\bx(i-1), \bS\bx(i-2), \ldots, \bS^{M-2}\bx(i-M+1)\Big], \label{eq:retaining signal}
	\end{equation} 
note that  only one shift is required at  time instant $i$ to produce the filtered signal $\by(i)$. This means that $y_k(i)$ can be computed  using  only  local  information available within the one-hop neighborhood of node $k$.  Let $\bZ(i)$ denote the $N\times M$ matrix given by:
	\begin{equation}
	\bZ(i) \triangleq \left[ \bx(i) ,\, \bS\bx(i-1),\ldots,\bS^{M-1}\bx(i-M+1) \right],
	\end{equation}
then model \eqref{eq:gfingt} can be written alternatively as:
	\begin{equation}
	\by(i)=\bZ(i)\boldh^o+\bv(i) \label{eq:gftleq}
	\end{equation}
From model \eqref{eq:gftleq}, sample $y_k(i)$ at node $k$ can be written as:
	\begin{equation}
	y_k(i)=  \bz_k^\top(i)\boldh^o+v_k(i), \label{eq:model}
	\end{equation}
where $\bz^\top_k(i)$ is the $k$-th row of $\bZ(i)$ given by:
	\begin{equation}
	\label{eq: z_k(i)}
	\bz_k(i)\triangleq \col \big\{ [\bx(i)]_k,\, [\bS\bx(i-1)]_k,\ldots,  [\bS^{M-1}\bx(i-M+1)]_k \big\}.
	\end{equation}
Observe  in~\eqref{eq:model} that each node shares the same filter coefficient vector $\boldh^o$. The corresponding graph filter~\eqref{eq:gf} is referred to as  \textit{node-invariant} graph filter.  A more flexible model was introduced in~\cite{segarra2017optimal}, and called a \textit{node-variant} graph filter. It allows  the filter coefficients to vary across nodes as follows:
\begin{equation}
\bH \triangleq \sum_{m=0}^{M-1}\diag (\boldh^{(m)}) \bS^m, \label{eq:nvgf}
\end{equation}
with $\boldh^{(m)} \in \mathbb{R}^N$. If $\boldh^{(m)}=h_m \One$ for all $m$, model \eqref{eq:nvgf} reduces to the node-invariant model \eqref{eq:gf}. Otherwise, each node applies different weights to the shifted signal $\bS^m \bx$. Then  $y_k(i)$ in \eqref{eq:model} can be re-written as:
\begin{equation}
y_k(i)=  \bz_k^\top(i)\boldh_k^o+v_k(i), \label{eq:datamodel}
\end{equation}
where  $\boldh_k^o \in \mathbb{R}^M$ is the filter coefficient vector at node $k$ collected into $\boldh^{(m)} $, i.e.,  $[\boldh_k^o]_m=[\boldh^{(m)}]_k$. In this work, we seek to estimate  $\boldh_k^o$  from the filtered graph signal $y_k(i)$ and  inputs  $\bz_k(i)$,  in a distributed, collaborative and adaptive manner. 
Distributed algorithms such as the diffusion LMS exist in the literature to address single-task and multitask inference problems with similar data models as \eqref{eq:model} or \eqref{eq:datamodel}. In this work, however, regressors $\bz_k(i)$ in \eqref{eq: z_k(i)} are raised from graph shifted signals. This paper aims to exploit the graph shift structure in the regression data and incorporate it into the formulation of the distributed algorithm -- see expression~\eqref{eq:DLMS} further ahead. In the sequel, first, we shall study the case where  the filter coefficients are common for all  nodes, i.e. $\boldh_k^o=\boldh^o,  \forall k\in \N$.  We  shall show how to estimate $\boldh^o$ from streaming data $\{\by(i), \bx(i)\}$ in a centralized way and then, in a distributed way. Next, we shall assume that there are clusters of nodes within the graph, and each node in the same cluster uses the same filter. This  model is called a \textit{hybrid node-varying} graph filter~\cite{gama2017convolutional}. We shall introduce an unsupervised clustering method to allow each node to identify which neighboring nodes it should collaborate with.

\subsection{Centralized solution}

Before introducing the distributed method, we first introduce the centralized solution.  Consider the data model~\eqref{eq:gftleq} and assume that $\{\by(i),\bx(i),\bv(i)\}$ are zero-mean jointly wide-sense stationary random processes. Estimating $\boldh^o$ from $\{\by(i), \bZ(i) \}$ can be performed by solving the following problem:
\begin{equation}
\boldh^o = \arg \min_{\boldh}J(\boldh),
\end{equation}
where $J(\boldh)$ denotes the mean-square-error criterion:
\begin{align}
J(\boldh)&=\expec \|\by(i)-\bZ(i) \boldh \|^2 \notag  \\
&= \expec\{\by^\top(i) \by(i) \} -2\boldh^\top\br_{Zy} + \boldh^\top \bR_Z \boldh, \label{eq:gcost}
\end{align}
and the $M \times M$ matrix $\bR_Z$ and the $M \times 1$ vector $\br_{Zy}$ are given by:
\begin{equation}
\bR_Z \triangleq \expec \{\bZ^\top(i)  \bZ(i) \} ,  \quad \br_{Zy} \triangleq \expec \{\bZ^\top(i)  \by(i)\}.\label{eq:RX}
\end{equation}
By setting the gradient vector of $J(\boldh)$  to zero, the optimal parameter vector $\boldh^o$ can be found by solving:
\begin{equation}
\label{eq: closed form solution}
\bR_Z \boldh^o=\br_{Zy}.
\end{equation}
It can be verified that the $(m,n)$-th entry of  $\bR_Z$ is given by:
\begin{equation}
[\bR_Z]_{m,n}=\tr\left((\bS^{m-1})^\top \bS^{n-1} \bR_x(m-n) \right)
\end{equation}
where $\bR_x(m)\triangleq \expec\{\bx(i)\bx^\top(i-m)\}$. The $m$-th entry of the vector $\br_{Zy}$ is given by:
\begin{equation}
[\br_{Zy}]_m=\tr\left( (\bS^{m-1})^\top \bR_{xy}(m) \right),
\end{equation}
with $\bR_{xy}(m)\triangleq \expec\{\by(i) \bx^\top(i-m) \}$ denoting cross correlation function, which is assumed independent of time $i$.	

Instead of solving~\eqref{eq: closed form solution}, $\boldh^o$ can be sought iteratively by using the gradient-descent method:
\begin{equation}
\boldh(i+1)=\boldh(i) +\mu\big[\br_{Zy} -\bR_Z\boldh(i)\big],
\end{equation}
with $\mu>0$ a small step-size. Since the statistical moments are usually unavailable beforehand, one  way is to replace them by the instantaneous approximations $\bR_Z\approx \bZ^\top(i)\bZ(i)$ and $\br_{Zy} \approx \bZ^\top(i)  \by(i)$. This yields the LMS graph filter:
\begin{equation}
\boldh(i+1)=\boldh(i) +\mu\bZ^\top(i) \big[\by(i) -\bZ(i)\boldh(i)\big]. \label{eq:centLMS}
\end{equation}
This stochastic-gradient algorithm  is referred to as the \textit{centralized graph-LMS} algorithm. In this centralized setting, each node $k$ at each time instant $i$ sends its data $\{x_k(i),y_k(i)\}$ to a fusion center which will update $\boldh(i)$ according to \eqref{eq:centLMS}. Note that the step-size $\mu$ in \eqref{eq:centLMS}  must satisfy $0<\mu<\frac{2}{\lambda_{\max}({\bR_Z})}$  in order to guarantee stability in the mean under certain independence conditions on the data~\cite{sayed2008adaptive}.

\section{Diffusion LMS Strategies over Graph Signals}\label{sec:diffusion}
In this section, we seek  to estimate the graph filter coefficients in a distributed  fashion.  First, we review the graph diffusion LMS strategy~\cite{nassif2017distributed}. Then, a preconditioned algorithm is proposed to improve the transient performance. We also devise a computationally efficient counterpart of this algorithm.
%----------------
	\subsection{Graph diffusion LMS}
	
Consider the local data model~\eqref{eq:model} at node $k$. It is worth noting that, by retaining the past shifted signals $\{[\bS^{m-1}\bx(i-m)]_{\ell}:m=1,\ldots,M-1\}$ {at each node $\ell$ in the network} from previous iterations, $\bz_k(i) $ can be computed locally at node $k$ from its one-hop neighbors at each iteration $i$. Let $\bR_{z,k}\triangleq\expec\{\bz_k(i)\bz^\top_k(i) \}$ denote the $M \times M$ covariance matrix with   $(m,n)$-th entry given by~\cite{nassif2017distributed}:
	\begin{equation}
	\label{eq: local covariance matrix}
	[\bR_{z,k}]_{m,n}=\tr\left([\bS^{m-1})]_{k,\bullet}^\top [\bS^{n-1}]_{k,\bullet} \bR_x(m-n) \right).
	\end{equation}
	
Considering the local cost $J_k(\boldh)$ at node $k$:
	\begin{equation}
	J_k(\boldh)=\expec |y_k(i)-\bz_k^\top(i) \boldh |^2, \label{eq:lcost}
	\end{equation}
the global cost \eqref{eq:gcost}   is now the aggregation of the local costs over the graph:
	\begin{equation}
	J(\boldh)=\sum_{k=1}^{N}	J_k(\boldh). \label{eq:cost}
	\end{equation}
In order to minimize~\eqref{eq:gcost} in a decentralized fashion, there are several useful techniques, e.g., incremental strategy~\cite{bertsekas1997new}, consensus strategy~\cite{xiao2005scheme} and diffusion strategy\cite{sayed2013diffusion}. Diffusion strategies are attractive since they are scalable, robust, and enable continuous learning and adaptation. In particular, the adapt-then-combine (ATC) diffusion LMS takes the following form at node $k$~\cite{nassif2017distributed}:
\begin{subequations}\label{eq:DLMS} 
		\begin{align}
		\bz_k^\top(i)&=\Big[ x_k(i), \sum_{\ell \in \N_k} s_{k\ell}~[\bz_\ell(i-1)]_{1},\ldots, \notag \\
		&\quad  \quad\quad \sum_{\ell \in \N_k} s_{k\ell}~[\bz_\ell(i-1)]_{M-1} \Big], \label{eq:DLMSz}\\ 
		\bpsi_k(i+1)&=\boldh_k(i)+\mu_k\bz_k(i)\big[y_k(i)-\bz_k^\top(i)\boldh_k(i)\big], \label{eq:DLMSa}\\
		\boldh_k(i+1)&=\sum_{\ell \in \N_k}a_{\ell k} \bpsi_\ell(i+1), \label{eq:DLMSb} 
		\end{align}
	\end{subequations}
where $\mu_k >0$ is a local step-size parameter and $\{a_{\ell k}\}$ are non-negative combination coefficients chosen to satisfy:
	\begin{equation}
	a_{\ell k} >0, \quad \sum_{\ell=1}^N a_{\ell k} =1, \, 
	\text{ and } \, a_{\ell k} =0 \, \text{ if }\,  \ell \notin \N_k. \label{eq:combcoeff}
	\end{equation}
	This implies that the matrix $\bA$ with $(\ell,k)$-th entry $a_{\ell k}$ is a left-stochastic matrix, which means that the sum of each of its columns is equal to $1$. In the first step~\eqref{eq:DLMSz},  each node $k$ uses the first $M-1$ entries of $\bz_{\ell}(i-1)$ from its one-hop neighbors and its own input sample $x_k(i)$ to compute $\bz_{k}(i)$. Note that the first $M-1$ entries of $\bz_{k}(i)$ then need to be retained for the next iteration. In the adaptation step \eqref{eq:DLMSa}, each node $k$ updates its local estimate~$\boldh_k(i)$ to an intermediate estimate $\bpsi_k(i+1)$. In the combination step \eqref{eq:DLMSb},  node $k$ aggregates all the intermediate estimates $\bpsi_\ell(i+1)$ from its neighbors to {obtain} the updated estimate $\boldh_k(i+1)$.
	%----------
	\subsection{Graph diffusion preconditioned LMS}
	
	The regressor  $\bz_k(i)$ used in the adaptation step~\eqref{eq:DLMSa} results from shifted graph signals while the shift matrix $\bS$ is not energy preserving in general~\cite{gavili2017shift}. This is due to the fact that the magnitude  of the eigenvalues of the shift operator $\bS$ are not uniformly equal to $1$; the energy of the shifted signal $\bS^m \bx$  changes exponentially with $m$. Thus, the eigenvalue spread of $\bR_{z,k}$ may be large and the LMS update may suffer from slow convergence speed in this case~\cite{sayed2008adaptive}. To address this issue, albeit at an increased computational cost, we resort to a form of Newton's method.  Focusing on the adaptation step, we have:
	\begin{equation}
	\bpsi_k(i+1)=\boldh_k(i)-\mu_k [\nabla_{\boldh}^2 J_k(\boldh_k(i))]^{-1} [\nabla_{\boldh}J_k(\boldh_k(i))],\label{eq:newton}
	\end{equation}
	where $\nabla_{\boldh}^2 J_k(\cdot)$ denotes the Hessian matrix for $J_k(\cdot)$ and $\nabla_{\boldh}J_k(\cdot)$ is its gradient vector, if available. For the quadratic cost function \eqref{eq:lcost}, expression~\eqref{eq:newton} would lead to:
	\begin{equation}
	\bpsi_k(i+1)=\boldh_k(i)+\mu_k\bR_{z,k}^{-1} 
	\big[\br_{zy,k}-\bR_{z,k} \boldh_k(i)\big], \label{eq:quanewton} 
	\end{equation}
	where $\br_{zy,k}=\expec\{\bz_{k}(i)y_k(i) \}$. Note that the second term on the RHS of \eqref{eq:quanewton} requires second-order moments. Since  they are rarely available beforehand, we can replace $\br_{zy,k}-\bR_{z,k} \boldh_k(i)$ by the instantaneous approximation:
	\begin{equation}
	\br_{zy,k}-\bR_{z,k} \boldh_k(i) \approx \bz_k(i) e_k(i) 
	\end{equation} 
	with $e_k(i)={y_k(i)}-\bz_k^\top(i)\boldh_k(i)$. 	The  adaptation step~\eqref{eq:quanewton} becomes:
	\begin{equation}
	\bpsi_k(i+1)=\boldh_k(i) +	\mu_k\widehat{\bR}^{-1}_{z,k}(i) \bz_k(i)e_k(i), \label{eq:LMSNadpt}
	\end{equation}
	where $\widehat{\bR}_{z,k}(i)$  is an estimate for ${\bR}_{z,k}(i)$ which can possibly be  obtained recursively: 
	\begin{equation}
	\widehat{\bR}_{z,k}(i)=(1-\mu)\,\widehat{\bR}_{z,k}(i-1) +\mu\big[\bz_k(i) \bz_k^\top(i)\big],\quad i\geq1, \label{eq:EstHessian}
	\end{equation}
	where $\mu$ is a small factor that can be chosen in $(0, 0.1]$ in practice. It can be verified that $ \expec \{\widehat{\bR}_{z,k}(i)\} = {\bR}_{z,k}$ is an unbiased estimate when $ i\rightarrow \infty$.	 As discussed before, $\bS $ may not be energy preserving and results in a large eigenvalue spread of $\bR_{z,k}$,  which may  even be close to singular. The inverse $\bR^{-1}_{z,k}$ would then be ill-conditioned and lead to undesirable effects. To address this problem, it is common to use regularization~\cite{sayed2008adaptive}. We obtain the diffusion LMS-Newton algorithm:
	\begin{subequations}
		\label{eq:DLMSN}
		\begin{align}
		\bpsi_k(i+1)&=\boldh_k(i) 
		+\mu_k\big[\epsilon\bI
		+\widehat{\bR}_{z,k}(i)\big]^{-1}\bz_k(i)e_k(i),\label{eq:DLMSNa} \\
		\boldh_k(i+1)&=\sum_{\ell \in \N_k}a_{\ell k} \bpsi_\ell(i+1),
		\end{align}
	\end{subequations}
    with $\epsilon \geq 0$ a small regularization parameter.  Compared with the diffusion LMS algorithm~\eqref{eq:DLMS}, algorithm~\eqref{eq:DLMSN} requires first to recursively estimate the Hessian matrix according to~\eqref{eq:EstHessian} and then calculate $\big[\epsilon \bI+\widehat{\bR}_{z,k}(i)\big]^{-1}$. This algorithm can lead to improved performance as shown in the sequel, but at the expense of additional computation cost.
	
	In order to reduce the computational complexity of the LMS-Newton algorithm, we propose to use a preconditioning matrix $\bP_k$ does not depend on the graph signal $\bx(i)$ in the adaptation step, instead of the Hessian matrix $\bR_{z,k}$ or its estimate $\widehat{\bR}_{z,k}$. Since the large eigenvalue spread of the input covariance matrix $\bR_{z,k}$ {results mainly} from the shift matrix $\bS$ and the filter order $M$, we construct an  $M \times M$ preconditioning matrix $\bP_k$ %from the local knowledge of $\bS$ and $M$ 
	as follows: 
	\begin{equation}
	\bP_k\triangleq \diag\{ \|  [\bS^{(m-1)}]_{k,\bullet}  \|^2  \}_{m=1}^M. 	\label{eq: preconditioned matrix} 
	\end{equation}
	The rationale behind~\eqref{eq: preconditioned matrix} is that, in the case where $\bx(i)$ is i.i.d. with variance $\sigma^2$, it follows from~\eqref{eq: local covariance matrix} that $\bR_{z,k}=\sigma^2\bP_k$. According to~\eqref{eq: preconditioned matrix}, matrix  $\bP_k$ does not depend on $\bx(i)$ and can be evaluated beforehand at each node~$k$ during an initial step. Each node $k$ only requires to know the edge weights in its $M$-hop neighborhood, which can be performed in a decentralized manner. Interestingly, $\bP_k$ is  a diagonal matrix, which means that the matrix product in the adaptation step does not require expensive matrix inversion. Following the same line of reasoning as for the Newton algorithm~\eqref{eq:DLMSN}, a regularization term $\epsilon\bI_{M}$ can be  added to $\bP_k$. This leads to:
	\begin{equation}
	\bD_k=(\epsilon \bI_{M} + \bP_k)^{-1}.
	\end{equation}
	We arrive at the following preconditioned graph diffusion LMS strategy:
	\begin{subequations}\label{eq:DPLMS}
		\begin{align}
		\bpsi_k(i+1)&=\boldh_k(i) +	\mu_k\bD_k \bz_k(i) e_k(i),\label{eq:DPLMSa} \\
		\boldh_k(i+1)&=\sum_{\ell \in \N_k}a_{\ell k} \bpsi_\ell(i+1).\label{eq:DPLMSb}
		\end{align}
	\end{subequations}
	At each iteration $i$, node $k$ uses the local information to update the intermediate estimate $\bpsi_k(i+1)$ in the adaptation step~\eqref{eq:DPLMSa}. Then, in the combination step~\eqref{eq:DPLMSb}, the intermediate estimates~$\bpsi_\ell(i+1)$ from the neighborhood of node $k$ are combined to get $\boldh_k(i+1)$.  Although the preconditioning matrix $\bD_k$ is not the true Hessian matrix, we prove in Section~\ref{sec:analy} that the algorithm converges to the optimal solution $\boldh^o$ provided that it is stable. 
	
\subsection{Comparison with the graph diffusion LMS}

We explain  how preconditioning with~\eqref{eq: preconditioned matrix} improves  performance. For comparison purposes, let us first focus on the adaptation step of  diffusion LMS. 
At each node $k$, the $m$-th entry of $\boldh_k(i)$ is updated as follows: 
	\begin{equation}
	[\bpsi_k(i+1)]_m=[\boldh_k(i)]_m+\mu_k\big[ \bz_k(i)e_k(i)\big]_m.
	\end{equation}
During the transient phase, the $m$-th entry $[\boldh_k(i)]_m$ exponentially converges to its optimal value with a time constant~\cite{sayed2008adaptive}:
	\begin{equation}
	\tilde{\tau}_{m} \approx \frac{1}{2\mu_k\lambda_{m}}
	\end{equation}
where $\lambda_{m}$ denotes the $m$-th eigenvalue of $\bR_{z,k}$.  Given $\mu_k$, the convergence rate of each entry of $\boldh_k(i)$ then depends on the corresponding eigenvalue. Disparity between entries increases as the eigenvalue spread defined as $\lambda_{\max}/\lambda_{\min}$ increases.

The preconditioning matrix $\bD_k$ is diagonal at each node~$k$, which means that the $m$-th entry of $\boldh_k(i)$ in \eqref{eq:DPLMSa}  converges to its optimal value with a time constant:
    \begin{equation}
	    \tilde{\tau}_{m} \approx \frac{1}{2\mu_k d_{k,m}\lambda_{m}}
	\end{equation}
where $d_{k,m}$ denotes the $m$-th diagonal entry of $\bD_k$.
The convergence speed now depends on $d_{k,m}\lambda_{m}$. Considering the case where $\bP_k$ is proportional to $\bR_{z,k}$, then $d_{k,m}$ is inversely proportional to $\lambda_m$, which mitigates the effects of the eigenvalues spread. We shall analyze and illustrate in Section~\ref{sec:analy} and Section~\ref{sec:simu}, respectively, how this preconditioning improves convergence speed in more general cases.

\section{Performance analysis}\label{sec:analy}
We shall now analyze the stochastic behavior of the diffusion preconditioned LMS (PLMS) algorithm~\eqref{eq:DPLMS} in the sense of mean and mean-square error. We introduce the  following weight error vectors at each node $k$:
	\begin{equation}
	\tilde{\boldh}_k(i)=\boldh^o-\boldh_k(i), \qquad \tilde{\bpsi}_k(i)=\boldh^o-\bpsi_k(i), 
	\end{equation}
and we collect them across the nodes into the network weight error vectors :
	\begin{align}
	\tilde{\boldh}(i) &\triangleq \col \{ \tilde{\boldh}_1(i), \tilde{\boldh}_2(i), \ldots, \tilde{\boldh}_N(i)\}, \\
	\tilde{\bpsi}(i) &\triangleq \col \{ \tilde{\bpsi}_1(i), \tilde{\bpsi}_2(i), \ldots, \tilde{\bpsi}_N(i)\}.
	\end{align}
We refer to \textit{mean stability} of the error vector $\tilde{\boldh}(i)$ if the limit superior of $\|\expec \tilde{\boldh}(i)\|$ is bounded. Furthermore, we will claim that the algorithm converges in the mean to the optimum if $\expec\tilde{\boldh}(i)$ converges to zero as $i$ tends to $+\infty$ regardless of the starting point. \textit{Mean-square stability} refers to the case where the superior limit of $\expec\|\tilde{\boldh}(i)\|^2$ is bounded.

Let us introduce the following $N\times N$ block matrices with individual entries of size $M \times M$:
	\begin{align}
	\cA &\triangleq\bA\otimes \bI_M,\\
	\cM &\triangleq \bdiag\{\mu_k\bI_M \}_{k=1}^{N},\\
	\cD &\triangleq \bdiag\{ \bD_k\}_{k=1}^{N}.
	\end{align}
The estimation error  in~\eqref{eq:DPLMSa} can be written as:
	\begin{equation}
	e_k(i)={y_k(i)}-\bz_k^\top(i)\boldh_k(i)=\bz_k^\top(i)\tilde{\boldh}_k(i)+v_k(i).
	\end{equation}
Subtracting $\boldh^o$ from both sides of~\eqref{eq:DPLMSa} and using the above relation, then stacking $\tilde{\bpsi}_k(i) $ across the nodes, we  obtain
	\begin{equation}
		\tilde{\bpsi}(i+1) =\left(\bI_{NM}-\cM \cD \cR_z(i)\right)\tilde{\boldh}(i)-\cM \cD \bp_{zv}(i), \label{eq:psierror}
	\end{equation}
where $\cR_z(i)$ is an $N \times N $ block matrix with entries of size $M \times M$ define as:
	\begin{equation}
	\cR_z(i) \triangleq \bdiag\{\bz_k(i)\bz_k^\top(i)\}_{k=1}^{N},
	\end{equation}
and $\bp_{zv}(i)$ is an $N \times 1 $ block column vector with entries of size $M \times 1$ given by:
	\begin{equation}
	\bp_{zv}(i)\triangleq {\col}\{\bz_k(i)v_k(i) \}_{k=1}^{N}.
	\end{equation}	
Subtracting $\boldh^o$ from both sides of~\eqref{eq:DPLMSb}, we obtain the block weight error vector:
	\begin{equation}
	\tilde{\boldh}(i+1)= \cA^\top \tilde{\bpsi}(i+1). \label{eq:herror}
	\end{equation}
Finally, combing~\eqref{eq:psierror} and~\eqref{eq:herror}, the network weight error vector $\tilde{\boldh}(i)$ of algorithm~\eqref{eq:DPLMS} evolves according to:
	\begin{equation}
	\tilde{\boldh}(i+1)=\cB(i)\tilde{\boldh}(i)-\cA^\top \cM \cD \bp_{zv}(i), \label{eq:errorrec}
	\end{equation}
with \begin{equation}
	\cB(i)= \cA^\top \Big(\bI_{NM}-\cM \cD \cR_z(i)\Big).
	\end{equation}
To proceed with the analysis, we  introduce the following assumption. 
\begin{assumption}[independent inputs]
		\label{assum: independence}
The inputs {$\bz_{k}(i)$} arise from a zero-mean random process that is temporally white with $\bR_{z,k} >0$.
\end{assumption}
A  consequence of Assumption~\ref{assum: independence} is that $\bz_{k}(i)$ is independent of $ \boldh_{\ell}(j)$ for all $\ell$ and $ j \leq i$. This independence assumption is not true in the current work. Two successive regressors $\bz_k$ involve common entries that cannot be statistically independent as in a conventional FIR implementation. However, when the step-size is sufficiently small, conclusions derived under this assumption tend to be realistic. For more details and discussions, see~\cite[Section~16.4]{sayed2008adaptive}.  Since this assumption helps to simplify the derivations without constraining the conclusions, it is widely used in the literature of adaptive filters and adaptive networks~\cite{sayed2013diffusion,sayed2008adaptive}.  We shall see  in Section~\ref{sec:simu}  that   the resulting expressions  match well the simulation results for sufficiently small step-sizes. 
	
\subsection{Mean-error behavior analysis}
	
Taking expectations of both sides of \eqref{eq:errorrec}, using the fact that $\expec \bp_{zv}(i)=0 $, and applying Assumption~\ref{assum: independence}, we find that the network mean error vector  evolves according to:
	\begin{equation}
	\expec 	\tilde{\boldh}(i+1)=\cB\,\expec \tilde{\boldh}(i), \label{eq:MeanEpc}
	\end{equation}
where:
	\begin{align}
	\cB &\triangleq\expec \cB(i)=\cA^\top(\bI_{NM}-\cM\cD\cR_z),\\
	\cR_z &\triangleq\expec \cR_z(i)= \bdiag\{\bR_{z,k}\}_{k=1}^{N}.
	\end{align} 
	\begin{theorem}[Convergence in the mean] \label{th:1}
Assume that data model~\eqref{eq:model} and Assumption~\ref{assum: independence} hold. Then, for any initial condition,  algorithm \eqref{eq:DPLMS} converges  asymptotically in the mean toward the optimal vector $\boldh^o$ if, and only if, the step-sizes in $\cM$ are chosen to satisfy:
		\begin{equation}
		\rho\left(\cA^\top(\bI_{NM}-\cM\cD\cR_z)\right)<1, \label{eq:stabilityrho}
		\end{equation} 
where $\rho(\cdot)$ denotes the spectral radius of its matrix argument.  In the case where the signal $\bx(i)$ is i.i.d, a sufficient condition for~\eqref{eq:stabilityrho}  to hold is to choose $\mu_k$ such that:
		\begin{equation}
		0<\mu_k<\frac{2}{\lambda_{\max}(\bD_k\bR_{z,k})}, \qquad k=1, \ldots, N. \label{eq:stability}
		\end{equation}
	\end{theorem}
	\begin{proof}
The weight error vector $\tilde{\boldh}(i)$ converges to zero if, and only if, the coefficient matrix $\cB$ in~\eqref{eq:MeanEpc} is a stable matrix, namely, $\rho(\cB)<1 $. Since any induced matrix norm is lower bounded by the spectral radius, we have the following relation in terms of  block maximum norm~\cite{sayed2013diffusion}:
		\begin{align}
		\rho(\cB) &\leq \|\cA^\top(\bI_{NM}-\cM\cD\cR_z) \|_{b,\infty} \notag\\
		&\leq\|\cA^\top \|_{b,\infty} \cdot \|\bI_{NM}-\cM\cD\cR_z \|_{b,\infty}  \notag\\
		&= \|\bI_{NM}-\cM\cD\cR_z \|_{b,\infty} , \label{eq:rhoBbound1}
		\end{align} 
		where the last equality follows from the fact that $\cA$ is left stochastic, which implies that $\|\cA^\top \|_{b,\infty} =1$ from Lemma D.4 of~\cite{sayed2013diffusion}. Matrix $\cR_{z}$ is block diagonal if $\bx(i)$ is i.i.d. Since $\cD$ is also  diagonal, their product is symmetric, and $\cM\cD\cR_z$ is a block diagonal symmetric matrix. Then, following  Lemma D.5 of~\cite{sayed2013diffusion}, its block maximum norm agrees with its spectral radius:
		\begin{eqnarray}
		\|\bI_{NM}-\cM\cD\cR_z \|_{b,\infty}=\rho( \bI_{NM}-\cM\cD\cR_z ). \label{eq:rhoBbound2}
		\end{eqnarray}
Combining \eqref{eq:rhoBbound1} and \eqref{eq:rhoBbound2}, we verify that condition \eqref{eq:stability}  ensures the stability of $\cB$.
	\end{proof}
\subsection{Mean-square-error behavior analysis}
	
We shall now study the mean-square-error behavior of algorithm~\eqref{eq:DPLMS}. Let $\bSig$ be any $NM\times NM$   positive semi-definite matrix that we are free to choose. The freedom in selecting $\bSig$  will allow us to derive  different   performance measures about the network and the nodes. We consider the  weighted mean-square error vector, i.e., $\expec  \|\tilde{\boldh}(i)\|_{\bSig}^2$, where $\|\tilde{\boldh}(i) \|_{\bSig}^2\triangleq \tilde{\boldh}^\top(i)  \bSig \tilde{\boldh}(i)$. From   Assumption~\ref{assum: independence} and  $\expec  \bp_{zv}(i)=\Zero $,  using~\eqref{eq:errorrec}, we obtain  the following variance relation:
	\begin{equation}
	\expec  \|  \tilde{\boldh}(i+1) \|_{\bSig}^2=\expec \|  \tilde{\boldh}(i) \|_{\bSig^\prime}^2+\expec \|  \cA^\top \cM \cD \bp_{zv}(i) \|_{\bSig}^2,\label{eq:MSErec}
	\end{equation}
where $\bSig^\prime \triangleq \expec\{\cB^\top (i) \bSig  \cB (i)\}$. Let $\bsig$ denote  the $(NM)^2 \times 1$ vector obtained by  vectorizing matrix $\bSig$, namely, $\bsig=\vc(\bSig)$.  With some abuse of notation, we shall use $\| \cdot \|^2_{\bsig}$  to also refer to the  quantity $\| \cdot\|^2_{\bSig}$ when it is more convenient.  Let  $\bsig^\prime=\vc (\bSig^\prime)$. Considering that $ \vc(\bU\bSig \bW) = (\bW^\top \otimes \bU)\bsig$, it can be verified that:
	\begin{equation}
	\bsig^\prime=\cF \bsig, \label{eq:bsigp}
	\end{equation}
where $\cF$ is the $(NM)^2 \times(NM)^2$ matrix given by
	\begin{align}
	\cF \triangleq & \expec\{ \cB^\top(i) \otimes \cB^\top(i)\} \notag \\
	= &\Big(\bI_{(NM)^2}-\bI_{NM} \otimes \cR_z^\top \cD\cM - \cR_z^\top \cD\cM \otimes \bI_{NM} \notag\\
	  & +\cO(\cM^2) \Big) (\cA\otimes \cA),  \label{eq:cFi}
	\end{align}
where $\cO(\cM^2)$ denotes $\expec \{\cR_z^\top(i)\cD\cM \otimes \cR_z^\top(i)\cD\cM \}$, which depends on the square of the step-sizes, $\{\mu_k^2\}$. While we can continue the analysis by taking this factor into account as was done in other studies \cite{sayed2003fundamentals}, it is sufficient for the exposition to focus on the case of sufficiently small step-sizes where terms involving higher powers of the step-sizes $\{\mu_k\}$ can be ignored. Following the same line of reasoning, for sufficiently small step-sizes $\{\mu_k\}$, $\cF$ can be approximated by:
	\begin{equation}
	\cF \approx \cB^\top \otimes \cB^\top. \label{eq:cFaprx}
	\end{equation}
The second term on the RHS of \eqref{eq:MSErec} can be written as:
	\begin{equation}
	\expec \|  \cA^\top \cM \cD \bp_{zv}(i) \|_{\bSig}^2=\tr(\bSig \cG), \label{eq:2RHS}
	\end{equation}
where
	\begin{align}
	\cG &\triangleq \cA^\top\cM \cD \cS \cD \cM \cA, \\
	\cS &\triangleq \expec\{\bp_{zv}(i)\bp^\top_{zv}(i) \}=\bdiag\{\bsig_{v,k}^2 \bR_{z,k} \}_{k=1}^{N}.
	\end{align}
Using the property  $\tr(\bSig\bW)=[\vc(\bW^\top)]^\top \bsig$, combining~\eqref{eq:bsigp} and~\eqref{eq:2RHS}, the variance relation~\eqref{eq:MSErec} can be re-written as:
	\begin{equation}
	\expec \|  \tilde{\boldh}(i+1) \|_{\bsig}^2=\expec \|  \tilde{\boldh}(i) \|_{\cF\bsig}^2+[\vc(\cG^\top)]^\top \bsig. \label{eq:var}
	\end{equation}
	\begin{theorem}[Stability in the mean-square]
Assume that data model~\eqref{eq:model} and Assumption~\ref{assum: independence} hold.  Algorithm \eqref{eq:DPLMS}  converges in the mean-square sense  if the matrix $\cF$ in~\eqref{eq:cFi} is stable. Assume further that the step-sizes are sufficiently small such that~\eqref{eq:cFaprx} is a reasonable approximation. In that case,  the stability of $\cF$ is ensured if $\cB$ is stable. 
	\end{theorem}
	\begin{proof}
Iterating~\eqref{eq:var} starting from $i=0$, we obtain
		\begin{equation}
		\expec \|  \tilde{\boldh}(i+1) \|_{\bsig}^2=\expec \|  \tilde{\boldh}(0) \|_{\cF^{i+1}\bsig}^2+ [\vc(\cG^\top)]^\top\sum_{j=0}^{i} \cF^j \bsig, \label{eq:varrlt}
		\end{equation}
with initial condition $\tilde{\boldh}(0)=\boldh^o-\boldh(0)$. Provided $\cF$ is stable, $\cF^i \rightarrow 0$ as $i \rightarrow \infty$, then the first item on the RHS of~\eqref{eq:varrlt}  converges to zero and the second item converges to a finite value.  The weighted mean-square error converges to a finite value as $i\rightarrow\infty$ which implies that the algorithm \eqref{eq:DPLMS} will converge in the mean-square sense  if  $\cF$  is stable. Under the  sufficiently small  step-sizes assumption where the higher-order terms of $\cF$ in~\eqref{eq:cFi}  can be neglected,  approximation~\eqref{eq:cFaprx} is reasonable.  The eigenvalues of  $\cF$ are all the products of the eigenvalues of $\cB$, which means that $\rho(\cF)=[\rho(\cB)]^2$. It follows that $\cF$ is stable if $\cB$ is stable. Therefore, when the graph signal $\bx(i)$ is i.i.d, according to Theorem~\ref{th:1}, condition~\eqref{eq:stability}  ensures mean-square stability of the  algorithm under the assumed approximation~\eqref{eq:cFaprx}. 
	\end{proof}
	\begin{theorem}[Network transient MSD] Assume sufficiently small step-sizes that ensure mean and mean-square stability. The network transient mean-square deviation (MSD) defined as  $\zeta(i)=\frac{1}{N} \expec\|\tilde{\boldh}(i) \|^2 $  evolves according to the following recursion for $i \geq 0$:
		\begin{align}
		\zeta(i+1) =\zeta(i)  + \frac{1}{N} \Big(& [\vc(\expec\{\tilde{\boldh}(0)\tilde{\boldh}^\top(0)\}) ]^\top (\cF-\bI_{(NM)^2})\notag\\
		&+[\vc(\cG^\top)]^\top\Big)\cF^{i}\vc(\bI_{NM}). \label{eq:MSD1}
		\end{align}	
	\end{theorem}
	\begin{proof}
Comparing \eqref{eq:varrlt} at time $i+1$ and $i$, $\expec \|  \tilde{\boldh}(i+1) \|_{\bsig}^2$ is related to $\expec \|  \tilde{\boldh}(i) \|_{\bsig}^2$ as follows:
		\begin{align}
		\expec \|&  \tilde{\boldh}(i+1) \|_{\bsig}^2=\expec \|  \tilde{\boldh}(i) \|_{\bsig}^2 +[\vc(\cG^\top)]^\top\cF^i \bsig+\notag\\
		&[\vc(\expec\{\tilde{\boldh}(0)\tilde{\boldh}^\top(0)\}) ]^\top (\cF-\bI_{(NM)^2})\cF^{i}\bsig.
		\end{align}
Substituting $\bsig$ by $\frac{1}{N}\vc(\bI_{NM})$ leads to \eqref{eq:MSD1}.
	\end{proof}
Although  expression~\eqref{eq:MSD1} gives a compact form of  the transient MSD model, it may not be practical to  use   since $\cF$  is of size $(NM)^2\times(NM)^2$ and may become  huge for large networks or high order filters. For example, in the   simulation Section~\ref{sec:simu}, we considered a network consisting of $N=60$ nodes and a graph filter of degree $M=5$. Matrix $\cF$ is of size  $90000\times90000$ and requires a prohibitive amount of computational time and memory space (about 60GB in that case). To tackle this issue, we make use of the following properties of the Kronecker product:
	\begin{align}
	\vc(\bX\bY\bZ)&=(\bZ^\top \otimes \bX) \vc(\bY) \\
	\tr(\bX\bY)&=\big( \vc(\bY^\top)\big)^\top \vc(\bX).
	\end{align}
This leads to:
	\begin{corollary}[Alternative network transient MSD expression]
		\begin{equation}
			\begin{split}
			\zeta(i+1)&=\zeta(i) +\frac{1}{N} \tr\Big(  \cB^i \cG (\cB^i)^\top  \\ 
			&+\tilde{\boldh}(0)\tilde{\boldh}^\top(0)\left((\cB^{i+1})^\top \cB^{i+1}    -(\cB^i)^\top\cB^i   \right) \Big).
			\end{split}
			\end{equation}
	\end{corollary}
While the update with $\cF$ has a computation complexity of order $\cO((NM)^2)$, using $\cB$ only requires matrix manipulations of order $\cO(NM)$.
	\begin{corollary}[Network steady-state MSD]
		Consider sufficiently small step-sizes to ensure mean and mean-square convergence. The network steady-state  MSD is given by
		\begin{equation}
		\zeta^\star=\frac{1}{N} [\vc(\cG^\top)]^\top (\bI_{(NM)^2}-\cF)^{-1}\ \vc(\bI_{NM}). \label{eq:ssmsd}
		\end{equation}
	\end{corollary}
	\begin{proof}
The network steady-state  MSD is defined as:
		\begin{equation}
		\zeta^\star=\lim\limits_{i\rightarrow \infty} \frac{1}{N} \expec\{\|\tilde{\boldh}(i) \|^2 \}.
		\end{equation} 
If $\cF$ is stable, we obtain from~\eqref{eq:var} as $i\rightarrow \infty$:
		\begin{equation}
		\lim\limits_{i\rightarrow \infty}\expec \|  \tilde{\boldh}(i) \|_{(\bI_{(NM)^2}-\cF)\bsig}^2=[\vc(\cG^\top)]^\top \bsig. \label{eq:msd}
		\end{equation}
We obtain \eqref{eq:ssmsd} by substituting  $\bsig$ in \eqref{eq:msd} by $\frac{1}{N}(\bI_{(NM)^2}-\cF)^{-1}\ \vc(\bI_{NM})$.
	\end{proof}
	
Following the same  line of reasoning  as for  the transient MSD model,  the steady-state MSD given by~\eqref{eq:ssmsd} can be equivalently expressed as:
	\begin{corollary}[Alternative  steady-state  MSD expression]
		\begin{equation}
		\zeta^\star=\frac{1}{N} \sum_{i=0}^{\infty} \tr\Big( \cB^i \cG  ( \cB^i)^\top  \Big)
		\end{equation}
	\end{corollary}
This expression is obtained by a series expansion of~\eqref{eq:ssmsd}.   In practice, a limited number of iterations can be used instead of the upper limit index $i\rightarrow \infty$ to obtain an  accurate result.	
%----------
\section{Unsupervised Clustering for hybrid node-varying graph filter}  \label{sec:multi}
In Section~\ref{sec:diffusion}, we investigated the scenario where the nodes in a graph share a common  filter coefficient vector. Now we extend this model to the more flexible case~\eqref{eq:datamodel}, which allows the filter coefficients to  vary across the graph. We further assume that the graph is decomposed into $Q$ clusters of nodes $\C_q$ and, within each cluster $\C_q$, there is a common filter coefficient vector $\boldh_q^o$ to estimate, namely,
	\begin{equation}
	\boldh_k^o=\boldh_q^o,\qquad\text{if } k\in\C_q.
	\end{equation}
	
We assume that there is no prior information on the clusters composition and that the nodes do not know which other nodes share the same estimation task. Applying the  algorithm~\eqref{eq:DPLMS} within this context may result a bias  due to aggregate intermediate estimates from different data models. To address this issue,   automatic network clustering  strategies may be used~\cite{chen2015diffusion,zhao2015distributed,plata2016unsupervised,khawatmi2017decentralized} in order to inhibit cooperation between nodes from different clusters. These methods are based on local stand-alone estimation strategies 	that may not be efficient for the current context. Basically, the polynomial form~\eqref{eq:gf} of graph filters does not make the 	estimation of the filter coefficients reliable enough for the higher degrees. In the following, we tackle this problem by devising a clustering strategy based on the PLMS.
	
First, we introduce the $N \times N$ instantaneous clustering matrix $\bE_i$, whose  $(\ell,k)$-th entry shows if    node $k$ believes at time $i$ that its neighboring node $\ell$ belongs to the same cluster or not, namely,
	\begin{equation}	
	[\bE_i]_{ \ell k }= \begin{cases}
	1, \;  \text{if}\, \ell \in\N_k \,\text{and} \, k \,  \text{believes that} \, \boldh^o_k=\boldh^o_\ell,\\
	0,\; \text{otherwise}.
	\end{cases}
	\end{equation} 
At each time instant $i$, node $k$  infers which neighbors  belong to its cluster based on the non-zeros entries of the $k$-th  column  of $\bE_i$. We collect these entries into  a set $\N_{k,i}$, so that node $k$  only combines the intermediate estimates  from  its neighbors in $\N_{k,i}$. Condition~\eqref{eq:combcoeff}  on the combination coefficients becomes:
	\begin{equation}
	a_{\ell k} >0, \quad \sum_{\ell=1}^N a_{\ell k} =1, \, 
	\text{ and } \, a_{\ell k} =0 \, \text{ if }\,  \ell \notin \N_{k,i}, \label{eq:combcoeff2}
	\end{equation} 
	
Since the clustering information is unknown beforehand, we propose to learn $\bE_i$ in an online way by computing a Boolean variable $b_{\ell k}(i)$ defined as follows:
	\begin{equation}
	b_{\ell k}(i)=\begin{cases}
	1,\quad \text{if} \,  \|\bpsi_\ell(i+1)-\boldh_k(i) \|^2\leq \beta,  \\ \label{eq:boolean}
	0,\quad \text{otherwise},
	\end{cases}
	\end{equation}
with $\beta>0$ a preset threshold. Variable $b_{\ell k}(i)$ is defined from the  $\ell_2$-norm distance between the estimates at two neighboring nodes. If this distance is smaller than the threshold $\beta$, the two nodes are then assigned to the same cluster. Note that the distance between $\bpsi_\ell(i+1)$ and $\boldh_k(i)$ is used in \eqref{eq:boolean}, instead of the distance between $\boldh_\ell(i+1)$ and $\boldh_k(i+1)$, in order to merge the learning and the clustering processes. Consider the learning process at any node $k$ defined in~\eqref{eq:DPLMS}. Information about clusters is used in the combination step~\eqref{eq:DPLMSb}, where $\N_k$ now denotes the neighboring nodes of $k$ that share the same estimation task as node~$k$. This information should be available as soon as possible in order to avoid estimation bias. Considering the distance between $\boldh_\ell(i+1)$ and $\boldh_k(i+1)$ to decide if nodes $k$ and $\ell$ are in the same cluster would allow to update the composition of sets $\N_k$ and $\N_\ell$ in the combination step~\eqref{eq:DPLMSb} used to calculate parameter vectors $\boldh_\ell(i+2)$ and $\boldh_k(i+2)$. This latency time can be shortened by considering the distance between $\bpsi_\ell(i+1)$ and $\boldh_k(i)$ right after the adaptation step~\eqref{eq:DPLMSa}, and using this information to define $\N_k$ in the combination step.

%Note that the computation of $b_{\ell k}(i)$ couples $\bpsi_\ell(i+1)$ and $\boldh_k(i)$ in order to merge the  learning and clustering processes into a single iterative algorithm. \cred{This process can be done right after the adaptation step, and the resulted cluster information then can be used in the combination step.}
	
This strategy usually fails if left as is, because the estimation of higher-order coefficients is not reliable enough and results in bad clustering performance. We propose to estimate this distance from $M_k$ principal components of the estimates, to be defined, of the estimates. Because $\bR_{z,k}$ cannot be reasonably used to perform a Principal Component Analysis (PCA) of the input data, as it is rarely available beforehand and would involve significant additional computational effort, we suggest instead using matrix $\bP_k$. As for the PLMS, the rationale behind this is that $\bR_{z,k}=\sigma^2\bP_k$ when $\bx(i)$ is i.i.d. with variance $\sigma^2$. Another interest lies in that $\bP_k$ is a diagonal matrix, which greatly simplifies calculations. Without loss of generality, consider that the diagonal entries of $\bP_k$ are in decreasing order. Projecting data onto the first $M_k$ principal axes then reduces to selecting their first $M_k$ entries and set the other entries to zero. Dimension $M_k$ can be determined by setting the ratio of explained variance to total variance to some desired level $\tau$  as follows:
	\begin{equation}
	\label{eq:inertia}
	\begin{split}
	\minimize ~&~ M_k\\
	\text{subject to }&~ \sum_{m=1}^{M_k} \hat{\pi}_{k,m} \geq \tau
	\end{split}
	\end{equation}
with $\hat{\pi}_{k,m}= [\bp_k]_m /\tr(\bP_k)$   an approximation  of the {\em proportion of total variance}~\cite{jolliffe_principal_2016}. In practice, we can use a predefined threshold $\tau \in [0.9,1) $  to decide how many entries should be retained. The Boolean variable~\eqref{eq:boolean} becomes:
	\begin{equation}
	\label{eq:booleannew}
	b_{\ell k}(i)=\begin{cases}
	1,\quad \text{if} \;   \frac{\|\bpsi^\prime_\ell(i+1)-\boldh^\prime_k(i) \|^2}{\|\boldh^\prime_k(i) \|^2}\leq \beta,  \\
	0,\quad \text{otherwise},
	\end{cases}
	\end{equation}
with $\bpsi^\prime_\ell(i+1)$ and $\boldh^\prime_k(i) $ the first $M_k$ entries of $\bpsi_\ell(i+1)$ and $\boldh_k(i) $ respectively. Compared to~\eqref{eq:boolean}, note that the distance in \eqref{eq:booleannew} that accounts for the similarity between $\bpsi^\prime_\ell(i+1)$ and $\boldh^\prime_k(i)$ has been normalized. We suggest to choose $\beta\in(0,0.01]$ to lower the false detection rate. To reduce noise effects, we further introduce a smoothing step:
	\begin{equation}
	t_{\ell k}(i)=\nu t_{\ell k}(i-1) +(1-\nu)b_{\ell k}(i),
	\end{equation}
where $t_{\ell k}(i)$ is a trust level, and $\nu$ is a forgetting factor in $(0, 1)$ to balance  the past  and present cluster assignments. Once the trust level $t_{\ell k}(i)$ exceeds a preset threshold $\theta$, which can be chosen in $[0.5,1)$, node $k$ concludes that node $\ell$ belongs to its cluster, that is, 
	\begin{equation}
	\label{eq:Ei}
	[\bE_i]_{ \ell k }=\begin{cases}
	1,\quad \text{if} \; t_{\ell k}(i)\geq \theta,  \\
	0,\quad \text{otherwise}.
	\end{cases}
	\end{equation} 
Based on $[\bE_i]_{ \ell k }$, each node $k$ determines  at each time instant $i$ those nodes $\ell$ that it believes they belong to the same cluster, updates the  combination coefficients according to~\eqref{eq:combcoeff2}, and finally combines the estimates from its neighbors with~\eqref{eq:DPLMSb}.

\section{Numerical results} \label{sec:simu}
	
\subsection{Experiment with i.i.d. input data}
We first considered a zero-mean i.i.d. Gaussian graph signal $\bx(i)$  with covariance   $\bR_x=\diag \{\sigma_{x,k}^2\}_{k=1}^N$. Variances $\sigma_{x,k}^2$ were randomly generated from the uniform distribution $\mathcal{U}(1,1.5)$. In this setting, the graph signal sample $x_k(i)$ was independent of $x_\ell(j)$ for all $\ell$ and $ j\leq i$. We assumed the linear data model~\eqref{eq:model}.  The graph filter order was set to $M=5$ and the coefficients $h_m^o$ were randomly generated from the uniform distribution $\mathcal{U}(0,1)$. Noise $\bv(i)$ was zero-mean Gaussian with covariance   $\bR_v=\diag \{\sigma_{v,k}^2\}_{k=1}^N$. Variances $\sigma_{v,k}^2$ were randomly generated from  the  uniform distribution $\mathcal{U}(0.1,0.15)$.
	\begin{figure}[!ht]
		\centering
		\includegraphics[scale=0.5]{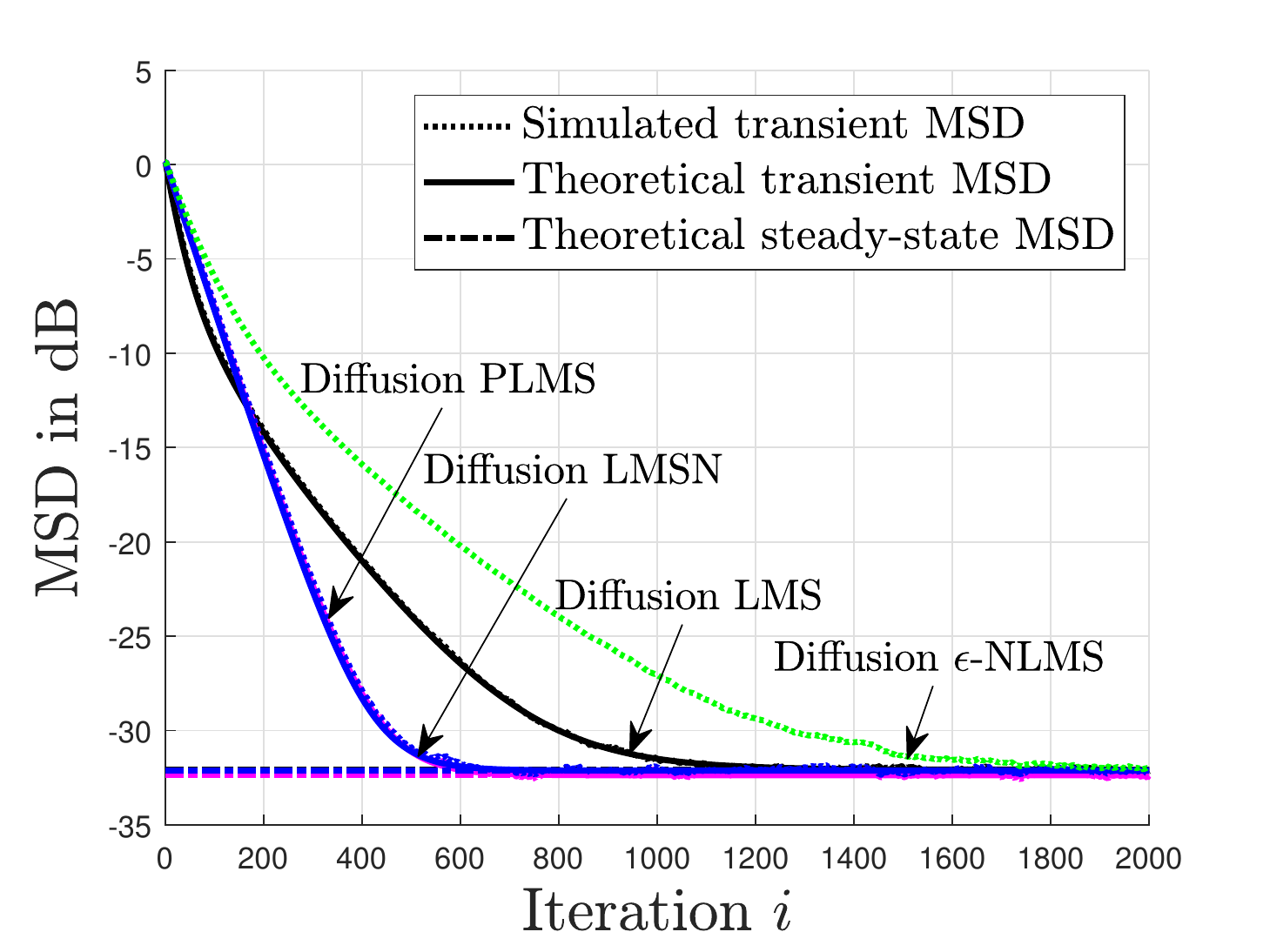}
		\caption{Network MSD performance with the Erd\H{o}s-R\'{e}nyi graph.}
		\label{fig_1}
	\end{figure}
	\begin{figure*}[!ht]
		\centering
		\subfigure[Normalized Adjacency Matrix]{
			\begin{minipage}[t]{0.3\linewidth}
				\centering
				\includegraphics[scale=0.4]{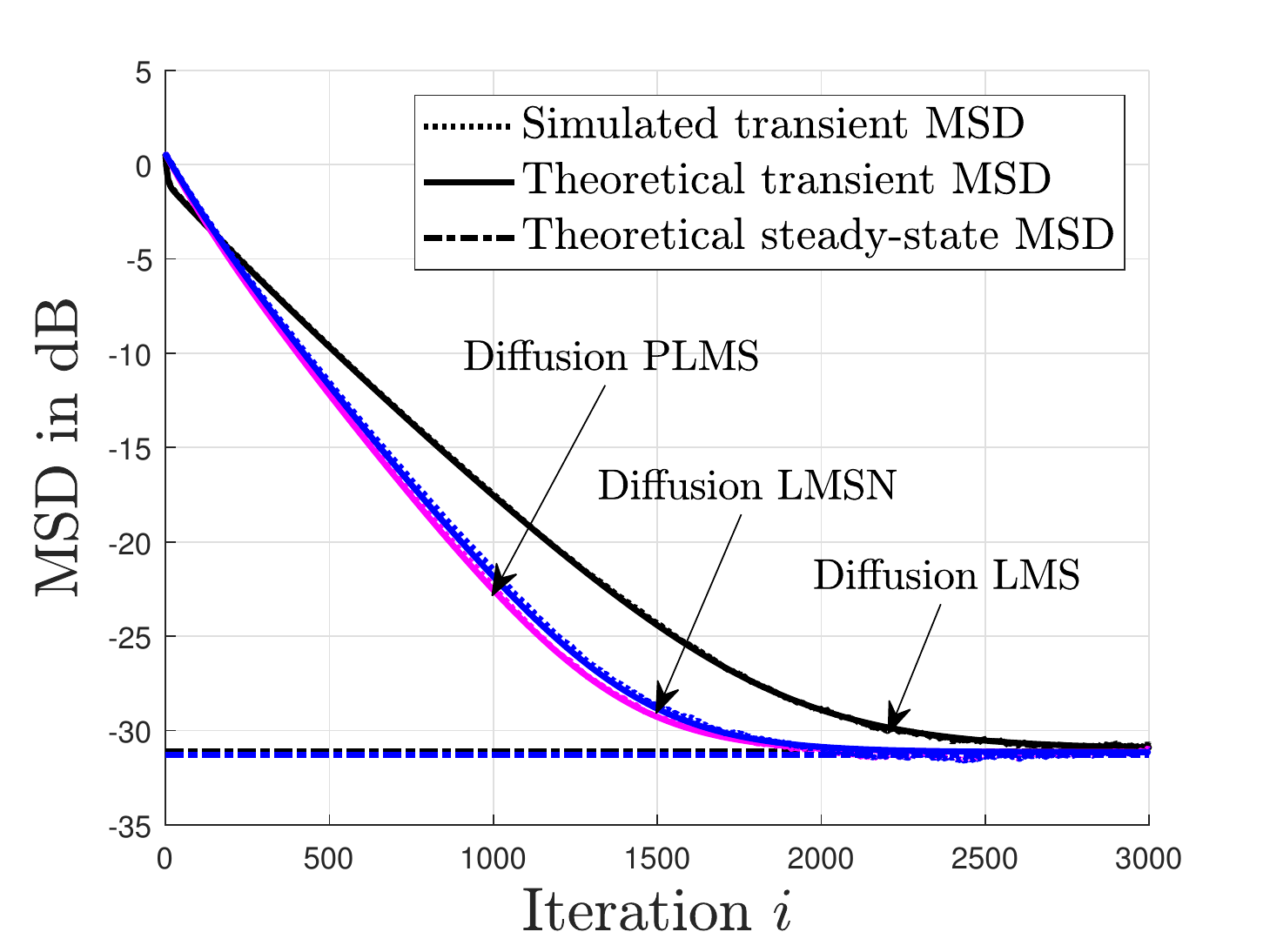}
			\end{minipage}%
		}%
		\subfigure[Normalized Laplacian Matrix]{
			\begin{minipage}[t]{0.3\linewidth}
				\centering
				\includegraphics[scale=0.4]{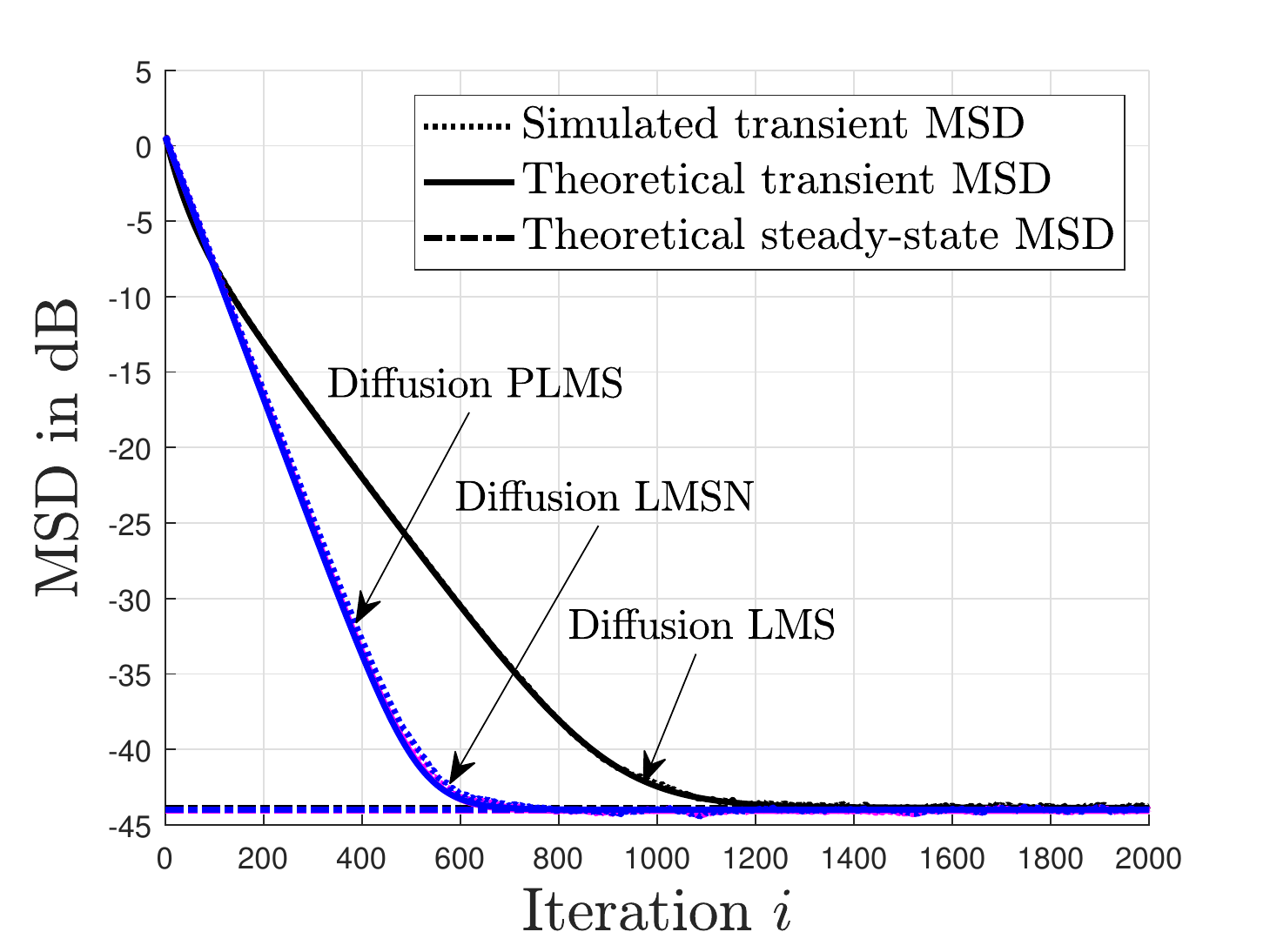}
			\end{minipage}%
		}%
		\subfigure[Adjacency Matrix]{
			\begin{minipage}[t]{0.3\linewidth}
				\centering
				\includegraphics[scale=0.4]{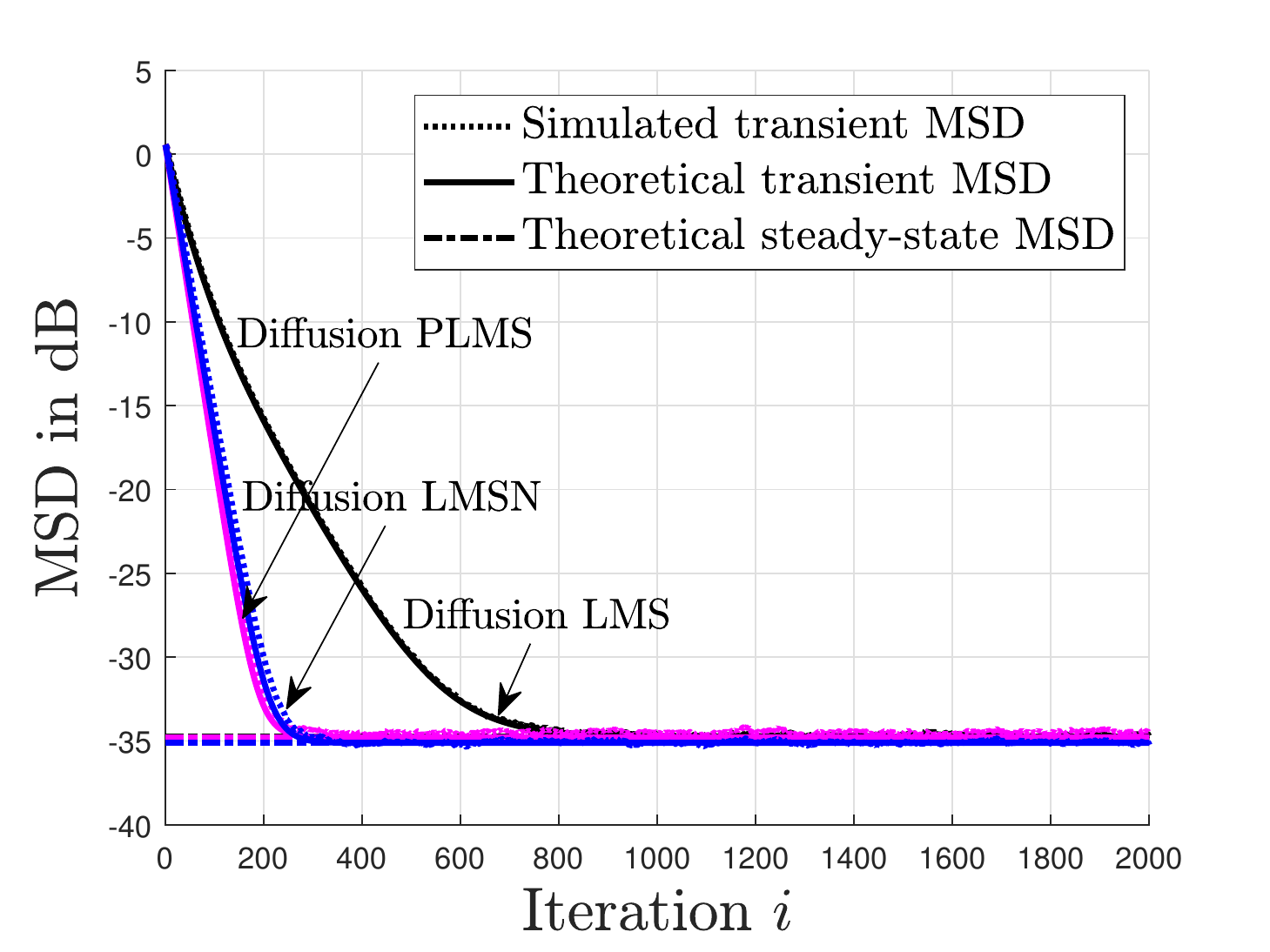}
			\end{minipage}%
		}%
		\caption{Network MSD performance for different types of shift operators with the sensor network.}
		\label{fig_2}
	\end{figure*} 
We considered this data model with an  Erd\H{o}s-R\'{e}nyi random graph and on a sensor network graph. Both  consisted of $N=60$ nodes. The Erd\H{o}s-R\'{e}nyi random graph was generated in a similar construction as in~\cite{mei2017signal}. Namely, it was obtained by generating an $N \times N$ symmetric matrix $\bS$ whose entries were governed by the Gaussian distribution $\N(0,1)$, and then thresholding edges to be between 1.2 and 1.8 in absolute value. Then, the edges were soft thresholded by 1.1 to be between 0.1 and 0.7 in magnitude. The shift matrix $\bS$ was normalized by 1.1 times its largest eigenvalue.  The  sensor network  was generated by using GSPBOX~\cite{perraudin2014gspbox}. Each node was connected to its 5 nearest neighbors.  The shift  matrix was the normalized adjacency matrix, that is,  $\bS=\frac{\bW}{1.1  \lambda_{\max}(\bW)}$. In this case, all the eigenvalues of $\bS$ are {smaller} than $1$ and the energy of the shifted signal $\bS^m\bx$ diminishes for large $m$. The smallest eigenvalue $\lambda_{\min}(\bR_{z,k})$ was very small, and, for some node, it was close to $0$. 

With this setting, we compared the diffusion LMS algorithm~\eqref{eq:DLMS}, the diffusion LMS-Newton (LMSN) algorithm~\eqref{eq:DLMSN}, and the diffusion preconditioned LMS (PLMS) algorithm~\eqref{eq:DPLMS}. Simulated results were averaged over 500 Monte-Carlo runs. For the LMSN and PLMS algorithms, we set the regularization parameter as $\epsilon=0.01$. We ran algorithms~\eqref{eq:DLMS},~\eqref{eq:DLMSN} and~\eqref{eq:DPLMS}  by setting  $a_{\ell,k}=\frac{1}{|\N_k|}$ for $\ell \in \N_k$.  We used a uniform step-size for all nodes, i.e., $\mu_k=\mu$ for all $k$. We also considered the $\epsilon$-normalized LMS ($\epsilon$-NLMS) method for comparison purposes. In this case, the adaptation step~\eqref{eq:DPLMSa} is substituted by:
    \begin{equation}
        \bpsi_k(i+1)=\boldh_k(i) +	\frac{\mu_k}{\|\bz_k(i)\|^2+\epsilon} \bz_k(i) e_k(i).
    \end{equation}
With the Erd\H{o}s-R\'{e}nyi graph, we compared the performance of the LMS, PLMS,  LMSN and $\epsilon$-NLMS algorithms. We set $\mu=\{0.08, 0.008, 0.01, 0.05\}$, respectively. The network MSD performance of each algorithm is reported in Fig.~\ref{fig_1}. The theoretical transient and steady-state MSD are also reported. With the sensor network graph, we compared the performance of the LMS, PLMS and LMSN algorithms. We set $\mu=\{0.08, 0.005, 0.0055\}$, respectively. The performance of each algorithm is reported in Fig.~\ref{fig_2}(a). In~Fig.~\ref{fig_1}, we observe that the diffusion $\epsilon$-NLMS converged slower than all other algorithms. Observe that the diffusion LMSN and PLMS algorithms converged faster than the LMS algorithm for both graphs, and the diffusion PLMS performed similarly compared with the LMSN in terms of convergence rate. Also, note that the theoretical results match well the simulated curves.
	 
In a second experiment, we considered the normalized graph Laplacian matrix   $\bS= \bD^{-\frac{1}{2}}\bL\bD^{-\frac{1}{2}}$, and the adjacency matrix {$\bW$},  as  shift operators.  For the normalized graph Laplacian, $\lambda_{\max}(\bR_{z,k})$ was large for all nodes. Therefore, for the diffusion LMS algorithm, the step-size was chosen relatively small to guarantee  convergence. We set $\mu=\{0.004,0.01,0.008\}$ for the LMS, LMSN   and PLMS. The results are reported in Fig.~\ref{fig_2}(b). For the adjacency matrix,   we used  uniform step-sizes $\mu=\{0.02,0.018\}$ for the LMSN   and PLMS, respectively. The step-size was set to $\mu_k=0.05 \cdot\frac{2}{\lambda_{\max}(\bR_{z,k})}$ for each node $k$ for the LMS update in order to achieve the same steady-state MSD. The results are reported in Fig.~\ref{fig_2}(c).  We observe in Fig.~\ref{fig_2} that the diffusion LMSN and PLMS algorithms converged faster than the LMS algorithm with the three graph shift operators. The PLMS algorithm achieved the same performance as the LMSN algorithm with a lower computational complexity.
	
\subsection{Experiment with correlated input data}
We tested the algorithms over the sensor network graph with correlated graph signals. We first considered a zero-mean i.i.d. Gaussian graph signal driven by a non-diagonal covariance matrix $\bR_x$. This means that the input data were correlated over the vertex domain, but uncorrelated over time. Matrix $\bR_x$ was generated as  $\bR_x=\bV\diag \{\sigma_{x,k}^2\}_{k=1}^N\bV^\top$, with  $\sigma_{x,k}^2$  randomly chosen from the uniform distribution $\mathcal{U}(1,1.5)$ and $\bV$ is the  graph Fourier transform matrix. The graph shift operator was defined by the normalized adjacency matrix. The filter degree was set as $M=3$. We observe in Fig.~\ref{fig_3} that the proposed diffusion PLMS algorithm performed as well as  the LMSN algorithm, and converged faster than the diffusion LMS algorithm. 
	\begin{figure}[!ht]
		\centering
		\includegraphics[scale=0.5]{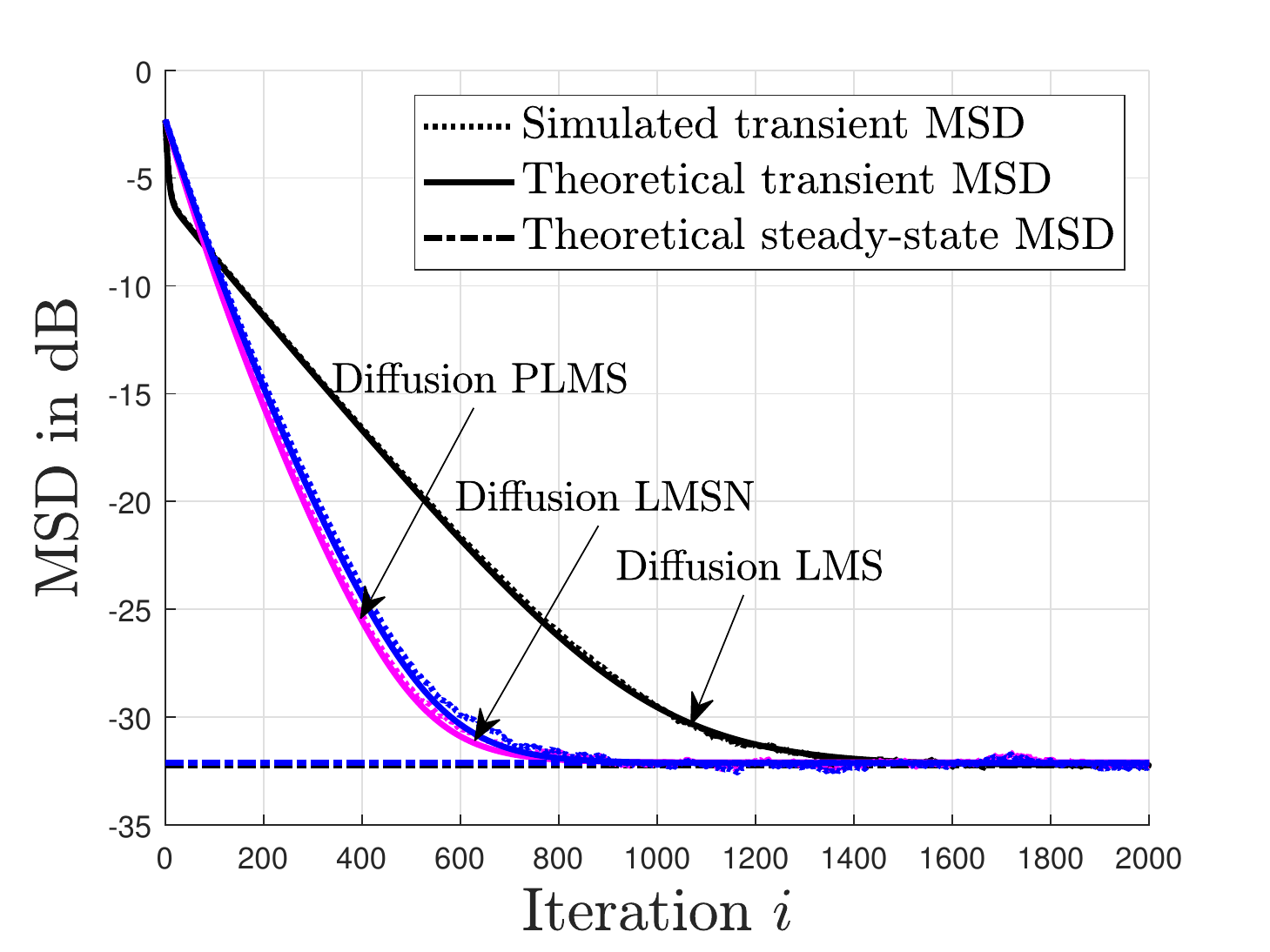}
		\caption{Network MSD performance with a vertex domain correlated input signal.}
		\label{fig_3}
	\end{figure}

Next, we considered a graph signal $\bx(i)$ that was correlated over both vertex and time domains. We assumed that  $\bx(i)$ is a Gaussian process with zero mean and covariance matrix $\bR_x$ satisfying the discrete Lyapunov equation: 
	\begin{equation}
	\bS \bR_x \bS^\top -\bR_x + \bI = 0. \label{eq:lyap}
	\end{equation}
Graph signal sample $\bx(i)$ was related to $\bx(i-1)$ as follows:
	\begin{equation}
	\bx(i) = \bS \bx(i-1)+\bw(i)
	\end{equation}
with $\bS$ the normalized adjacency matrix and $\bw(i)$ a zero-mean i.i.d.  Gaussian noise with covariance $\bI_N$. It can be checked that $\bx(i)$ is wide-sense stationary with $\expec \{\bx(i)\} = 0$, and $\bR_x(\tau) = \bS^\tau \bR_x(0)$ for all $\tau >0$, where $\bR_x(0)$ satisfies the Lyapunov equation~\eqref{eq:lyap}. The graph filter order was set as $M=3$. The step-sizes were set to $\mu=\{0.1,0.038,0.03 \}$ for the LMS,  LMSN and PLMS, respectively. The regularization parameter $\epsilon$  was set to $\epsilon=0.1$. Fig.~\ref{fig_4} depicts the simulated and  theoretical MSD performance.  We observe that, due to correlation over time, the diffusion LMSN method converged faster than the PLMS.  However,  the proposed PLMS algorithm still performed better than the diffusion LMS method. 	\begin{figure}[!ht]
		\centering
		\includegraphics[scale=0.5]{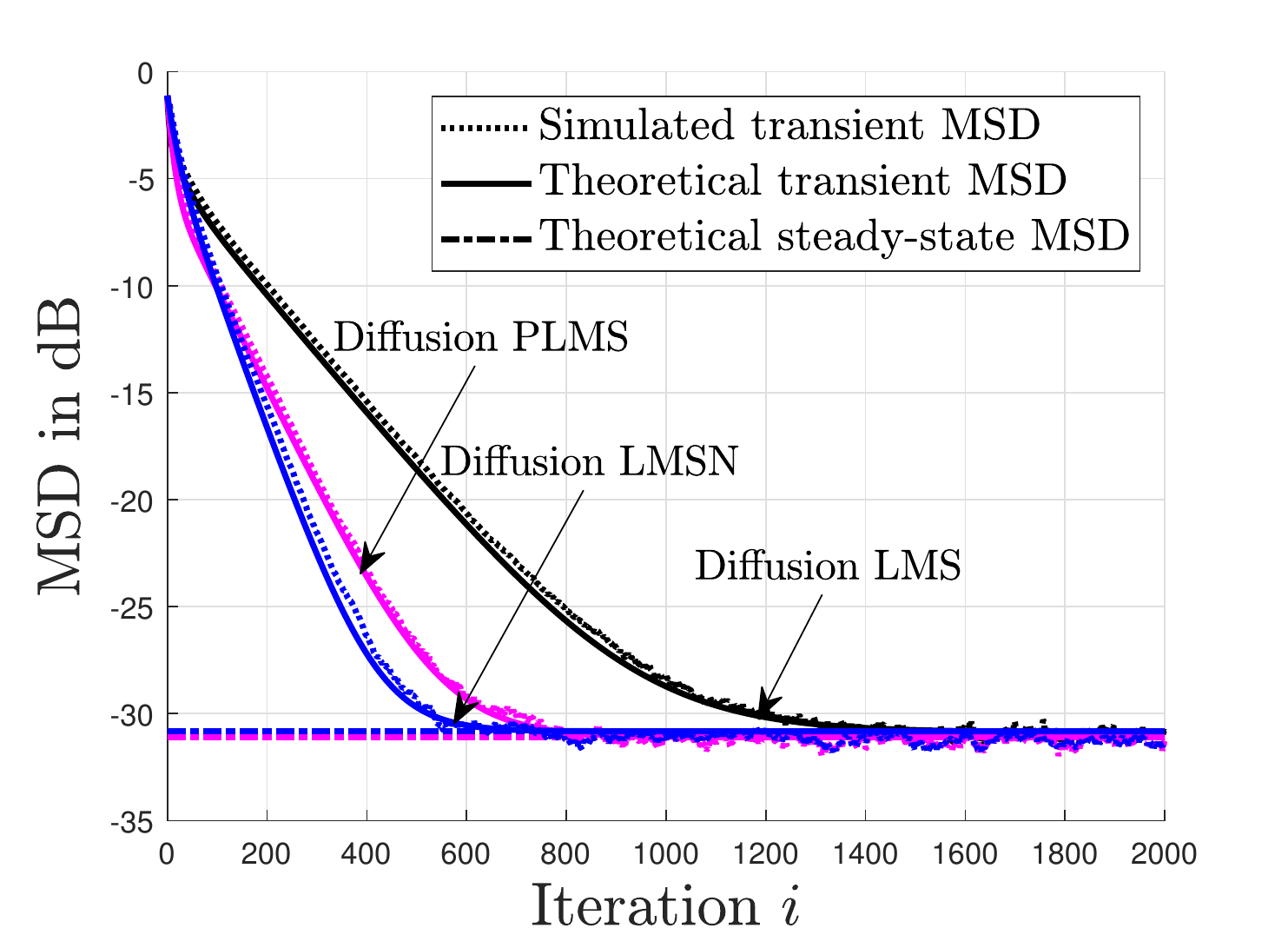}
		\caption{Network MSD performance with input graph signal correlated over both vertex and time domains.}
		\label{fig_4}
	\end{figure}
%%%%%%%%
\subsection{Clustering method for node-varying graph filter}
	
Finally, we considered a scenario where nodes do not share the same filter coefficients. We assumed the linear data model~\eqref{eq:datamodel}. The graph shift operator was defined by the normalized adjacency matrix, and the graph filter order was set as $M=3$. The nodes were  grouped into  three clusters: $\C_1=\{1,\ldots,20\}$, $\C_2=\{21,\ldots,40\}$, and $\C_3=\{41,\ldots,60\}$. The optimal graph filter coefficients $\boldh^o_k$ were set according to $[0.5 \; 0.4 \; 0.9 ]^\top$ if $k\in\C_1$, $[0.3\; 0.1 \;0.4 ]^\top$ if $k\in\C_2$, and $[0.9\; 0.3 \; 0.7]^\top$ if $k\in\C_3$.  We considered for comparison purpose the PLMS algorithm with clustering mechanism \eqref{eq:inertia}-\eqref{eq:Ei} , with basic clustering mechanism \eqref{eq:boolean} and $M_k=M$ for all $k$, the oracle PLMS algorithm where the clusters are assumed to be known a priori, the PLMS algorithm without clustering mechanism, and the non-cooperative algorithm where  $a_{\ell k}=1$ if $k=\ell$ and zero otherwise.  All algorithms used the adaptation step~\eqref{eq:DPLMSa} with the same step-size $\mu_k=0.01$ for all $k$. Parameters $\{\tau,\beta,\theta,\nu \}$ were set to $\{0.9, 0.01, 0.5, 0.98\}$, respectively. As shown in~Fig.~\ref{fig_5}, the non-cooperative method did not achieve acceptable MSD level. The main reason is that, with the normalized adjacency matrix as graph shift operator $\bS$, the entries of $\bz_k(i)$ in~\eqref{eq: z_k(i)} corresponding to higher powers of $\bS$ are significantly diminished, resulting in poor estimation performance of filter coefficients when nodes cannot cooperate. The PLMS algorithm without clustering mechanism did not achieve good performance too because it has been designed to converge toward a consensual solution $\boldh^o$, which does not make sense for this scenario. The PMLS with clustering mechanism (71) and $M_k = M $ achieved slightly improved performance because the estimation of higher-order coefficients in $\boldh_k$ was not reliable enough, leading to incorrect clustering. The proposed PLMS algorithm with clustering mechanism \eqref{eq:inertia}-\eqref{eq:Ei} performed as well as the oracle algorithm. Fig.~\ref{fig_6}~(a) shows the topology of the graph given by the adjacency matrix $\bA$ (and the shift matrix $\bS$). Fig.~\ref{fig_6}~(b) presents the clusters inferred by the proposed method. These clusters perfectly match the ground truth clusters $\C_1$ to $\C_3$.
	\begin{figure}[!ht]
		\centering
		\includegraphics[scale=0.3]{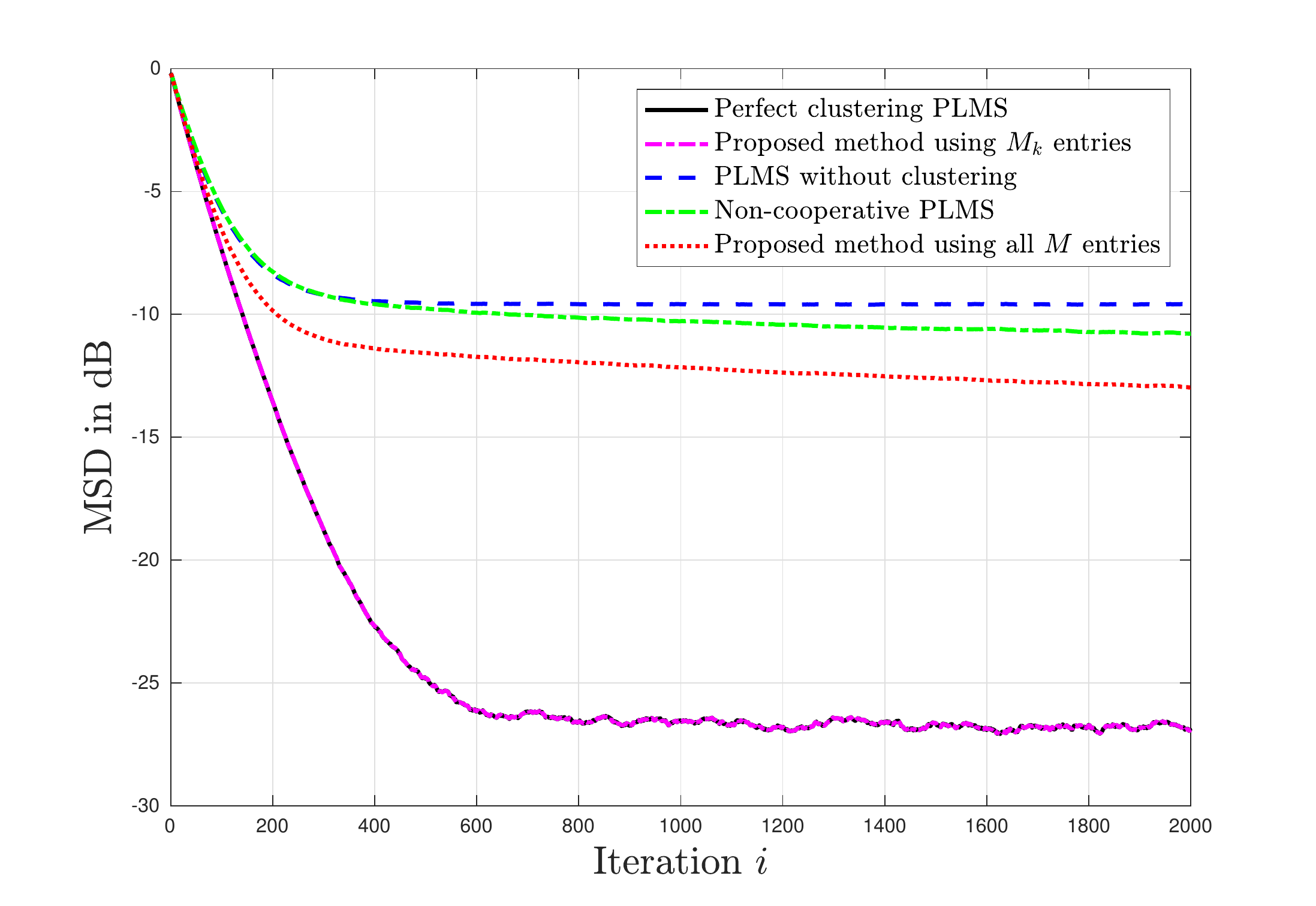}
		\caption{Network MSD performance for different clustering algorithms. }
		\label{fig_5}
	\end{figure}

	\begin{figure}[!ht]
		\centering
		\subfigure[Adjacency Matrix]{
			\begin{minipage}[t]{0.45\linewidth}
				\centering
				\includegraphics[scale=0.2]{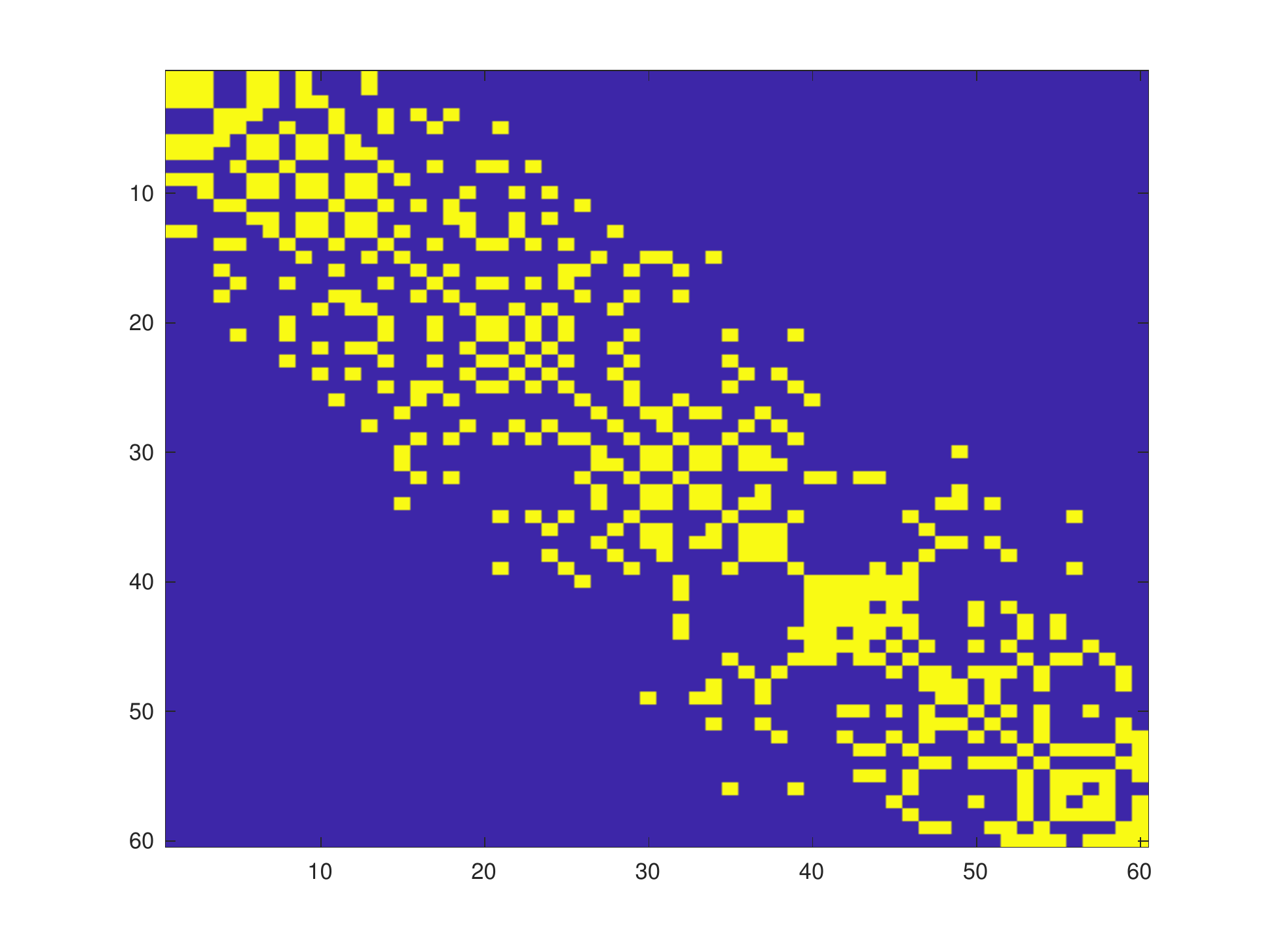}
			\end{minipage}%
		}%	
		\subfigure[Inferred clusters]{
		\begin{minipage}[t]{0.45\linewidth}
			\centering
			\includegraphics[scale=0.2]{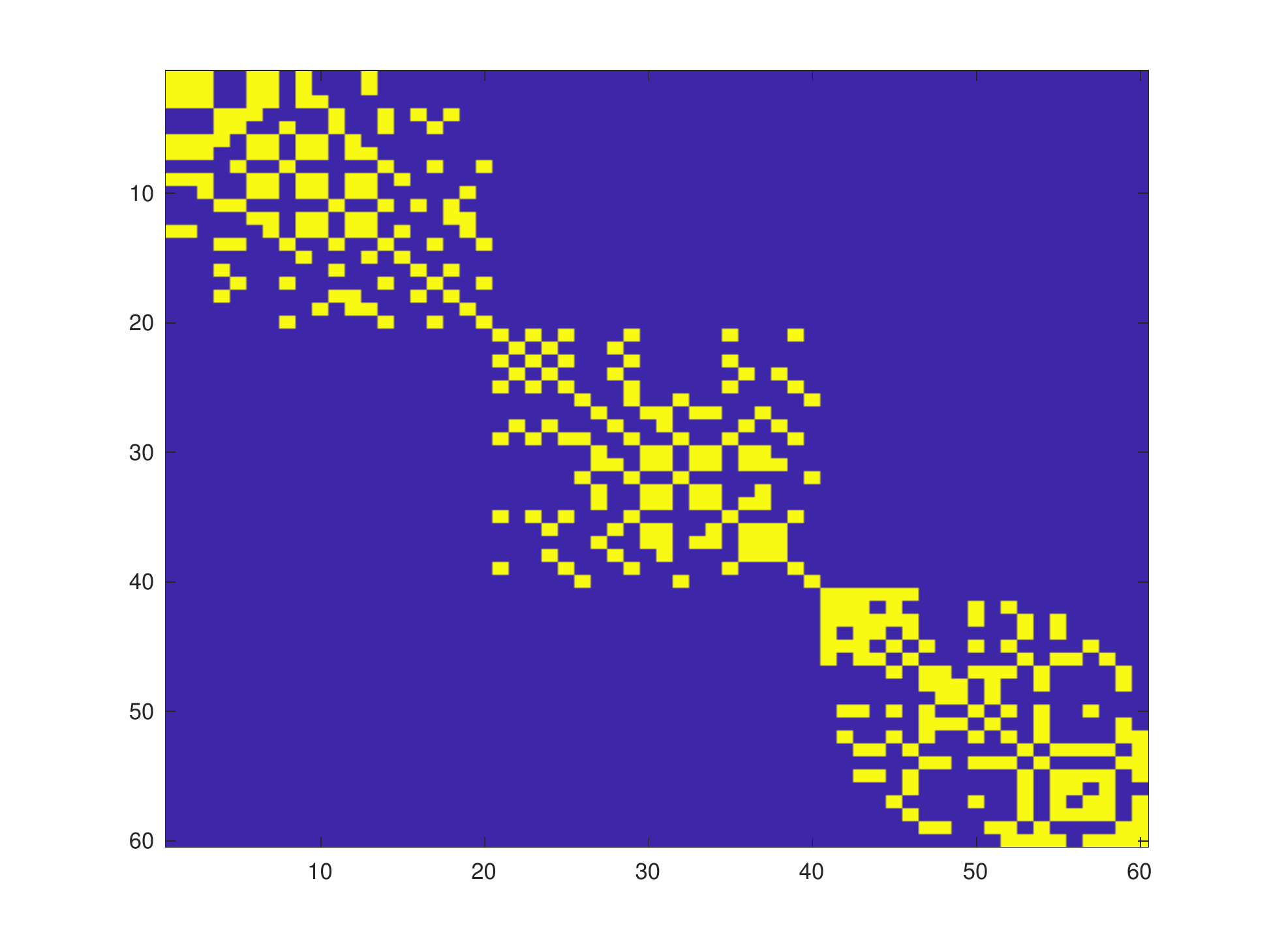}
		\end{minipage}%
		}%		
		%\centering
		%\subfigure[Adjacency Matrix]{\includegraphics[scale=0.3]{topology}} %
		%\subfigure[Inferred clusters]{\includegraphics[scale=0.3]{cluster}} 
		\caption{ Graph topology and clusters. }
		\label{fig_6}
	\end{figure}
	
Next, we considered that optimal parameter vectors $\boldh^o_k$  change over time while clusters remain unchanged. Nodes were grouped into two clusters $\C_1=\{1,\ldots,30\}$ and $\C_2=\{31,\ldots,60\}$. The optimal parameter vectors changed for both clusters at time instant $i=1000$. Simulation results in Fig.~\ref{fig_7} show  that the proposed clustering method was able to track well this change. Finally, we considered the scenario where clusters and models change simultaneously. At Stage 1, the nodes were grouped into two clusters defined as $\C_1=\{1,\ldots,30\}$ and $\C_2=\{31,\ldots,60\}$. Stage 2 started at time instant $i=1000$ with three clusters  $\C_1=\{1,\ldots,20\}$, $\C_2=\{21,\ldots,40\}$, and $\C_3=\{41,\ldots,60\}$. Stage 3 started at time instant $i=2000$ with two clusters $\C_1=\{1,\ldots,25\}$, $\C_2=\{26,\ldots,60\}$. At each stage, the optimal parameter vectors $\boldh_k^o$ changed accordingly. Ground truth clusters for the three stages are depicted in Fig.~\ref{fig_9}~(Top). Parameters $\{ \mu, \tau, \beta, \theta, \nu \}$ were set to $\{0.01, 0.9, 0.01, 0.5, 0.4\}$, respectively. Fig.~\ref{fig_8} shows the simulated transient MSD of the proposed PLMS algorithm with clustering mechanism. It is compared with the oracle PLMS algorithm where the clusters are assumed to be known a priori. Fig.~\ref{fig_9}~(Bottom) depicts the inferred clusters at $i=1000,2000,3000$ during one Monte Carlo run. The proposed algorithm was able to track changes in both clusters and models. 
	\begin{figure}[!ht]
		\centering
		\includegraphics[scale=0.3]{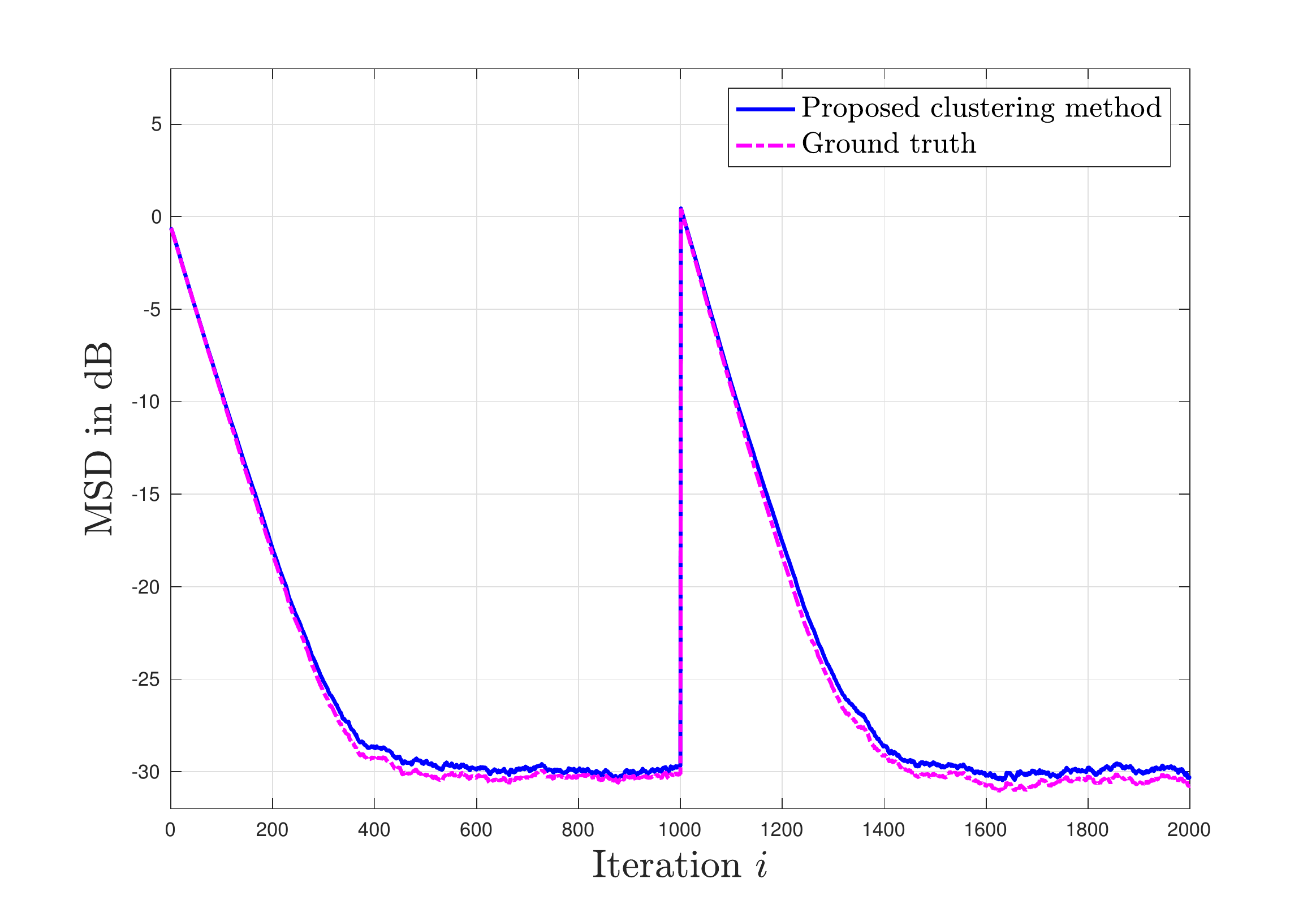}
		\caption{Network MSD performance with model change.}
		\label{fig_7}
	\end{figure}
	\begin{figure}[!ht]
		\centering
		\includegraphics[scale=0.3]{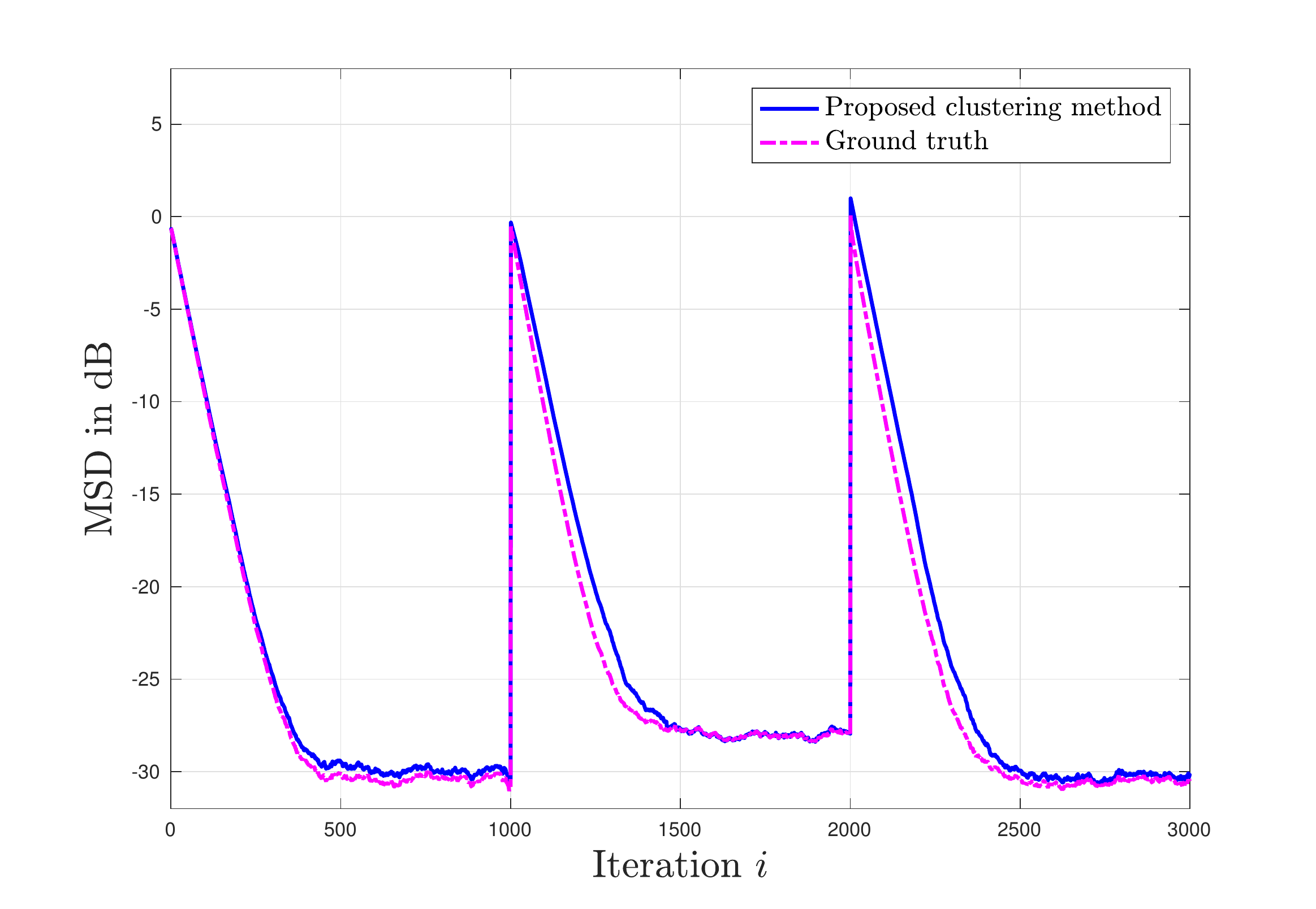}
		\caption{Network MSD performance with model  and clusters change.}
		\label{fig_8}
	\end{figure}
	\begin{figure}[!ht]
		\centering
		\includegraphics[scale=0.32]{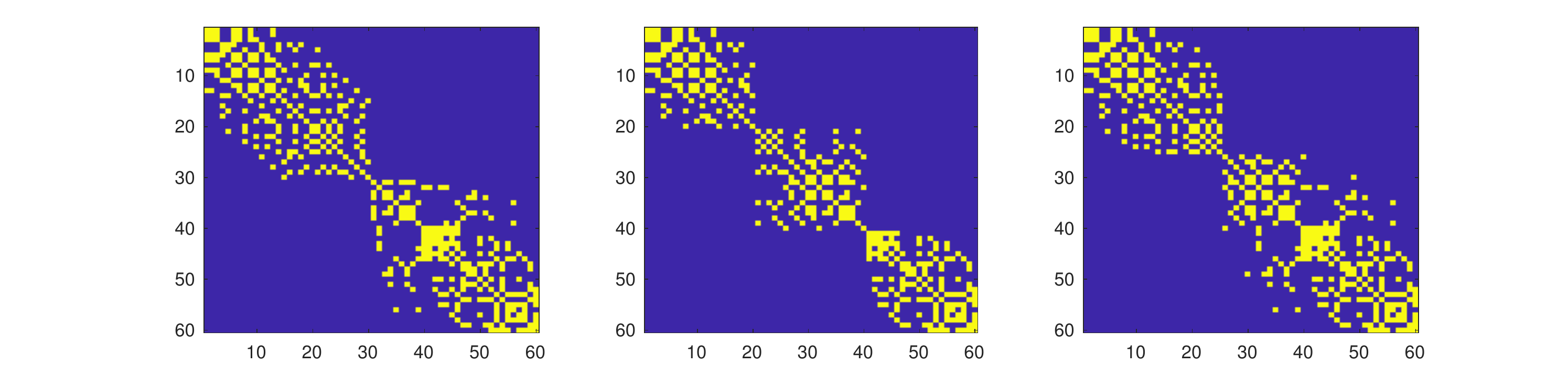}\\
		\includegraphics[scale=0.32]{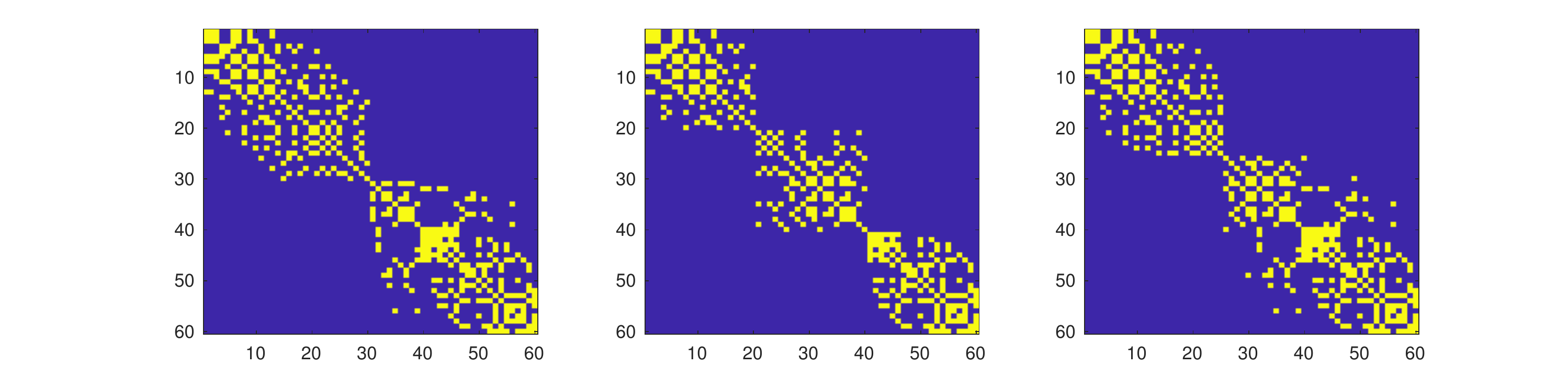}
		\caption{\textit{(Top)} Ground truth cluster. \textit{(Bottom)} Inferred clusters at steady-state of a single Monte Carlo run. From left to right: Stage~1, Stage~2, Stage~3. }
		\label{fig_9}
	\end{figure}

\subsection{Reconstruction on U.S. temperature dataset}
We considered a dataset that collects hourly temperature measurements at $N = 109$ stations for  $T =8759$ hours across the United States in 2010~\cite{noaa}. An undirected graph, illustrated in Fig.~\ref{fig_10}, was constructed according to the nodes coordinates by using the $k$-NN approach ($k=7$) of GSPBOX. 
\begin{figure}[!ht]
	\centering
	\subfigure[]{
		\begin{minipage}[t]{0.45\linewidth}
			\centering
			\includegraphics[scale=0.3]{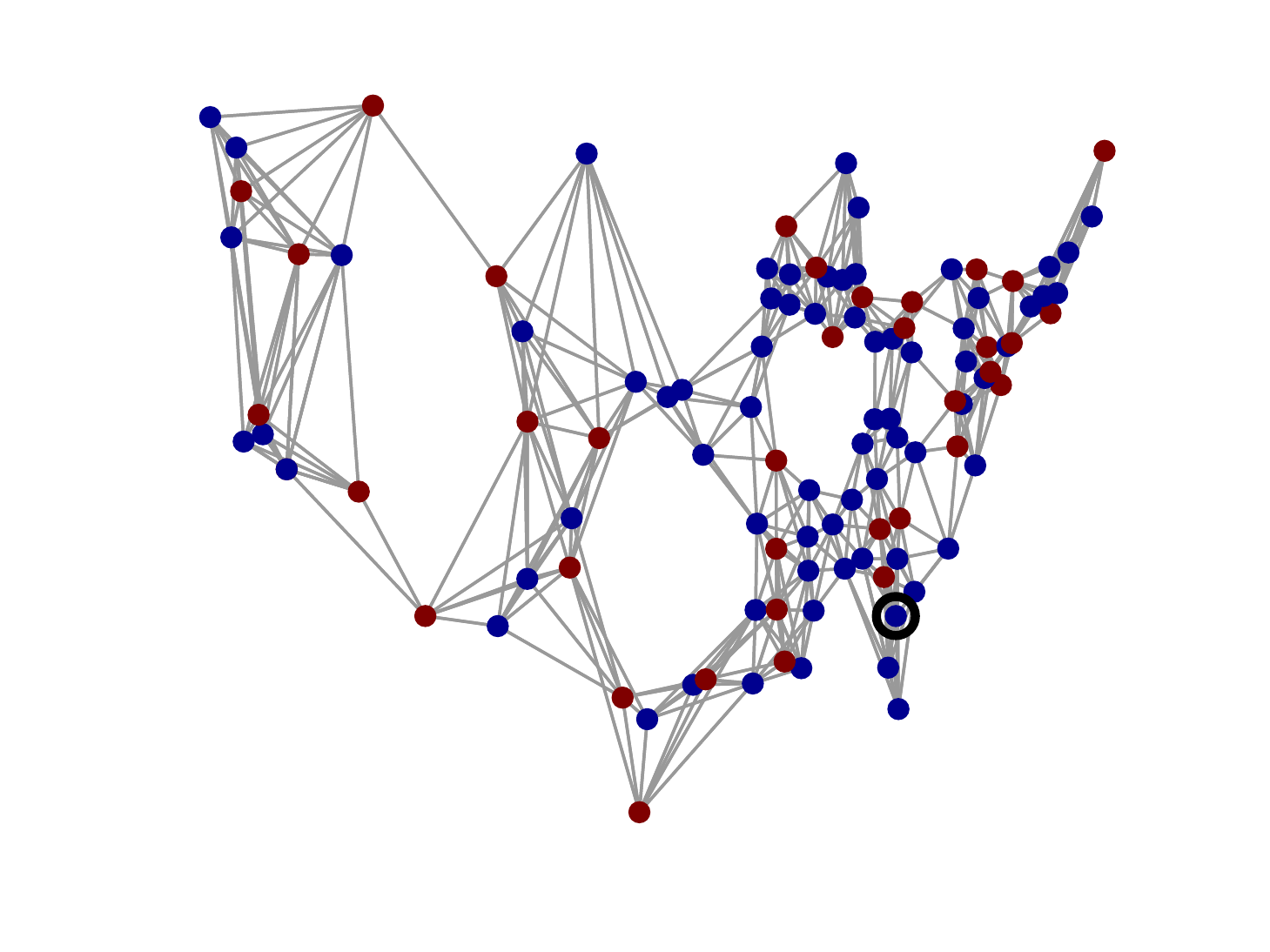}
		\end{minipage}%
	}%	
	\subfigure[]{
		\begin{minipage}[t]{0.45\linewidth}
			\centering
			\includegraphics[scale=0.3]{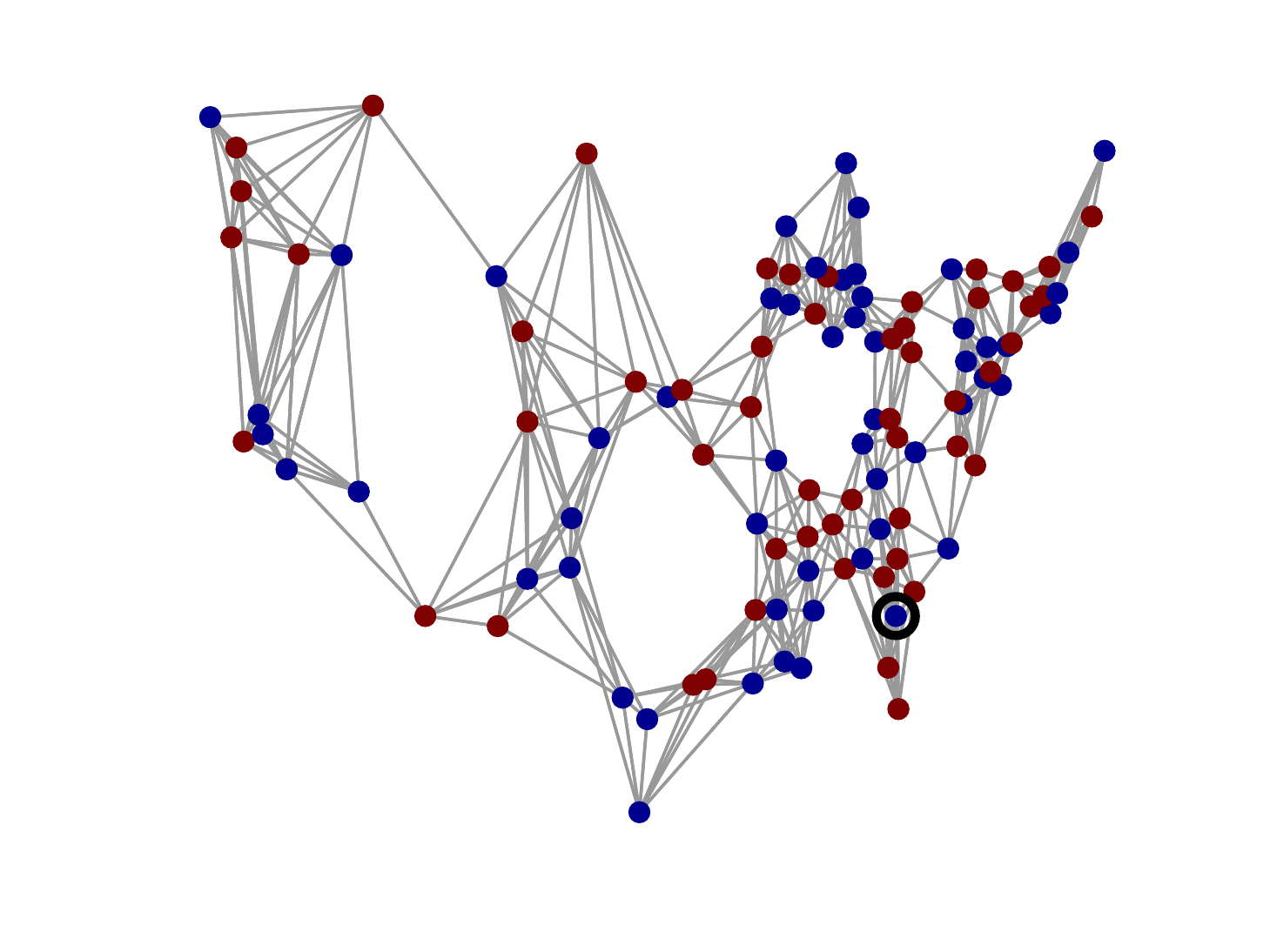}
		\end{minipage}%
	}%		
	\caption{Graph topology for the U.S. temperatures dataset. Temperatures were sampled at the red nodes in red. Data at the blue nodes were unobserved. (a)  37 sampled nodes. (b)  54 sampled nodes.}
	\label{fig_10}
\end{figure}

In the first experiment, the dataset was divided into a training set containing $T_{train}=6570$ hours data (about 75\% of total). The remaining data were assigned to the test set.  The goal of this experiment was to learn a graph filter that minimizes the reconstruction error over the training set, i.e.,
\begin{equation}
\min \sum_{i=1}^{T_{train}} \sum_{k=1}^N |y_k(i)-\sum_{m=1}^M h_{m,k} [\bS^m \bx(i-m+1)]_k  |^2, \label{eq:rd_prob}
\end{equation}
where $\by(i)$ is the ground truth temperature at time $i$, and $\bx(i)$ is the partial observation given by $\bx(i)=\diag(\One_\mathcal{S}) \by(i)$. Here $\One_\mathcal{S}$ denotes the set indicator vector, whose $k$-th entry is equal to one  if node $k$ is sampled, and zero otherwise. The sampling set, illustrated in Fig.~\ref{fig_10} (a) was fixed over time in the first experiment. The normalized adjacency matrix was set as graph shift operator. Graph filter degree was set to $M=4$. Note that if $h_{m,k}=h_m$ for all $k$, problem~\eqref{eq:rd_prob} refers to the single-task problem where all the nodes seek to find common graph filter coefficients; see model~\eqref{eq:gf}. Otherwise, problem~\eqref{eq:rd_prob} refers to the multitask problem; see model~\eqref{eq:nvgf}. We ran different models and algorithms on the training set to learn graph filter coefficients.   In Fig.~\ref{fig_11}, we provide the true temperature and the reconstructed ones obtained by the different algorithms at an unobserved node, black circled in Fig.~\ref{fig_10}, over the last 120 hours samples of the test set. For comparison purposes, reconstruction results of the Kernel Kalman Filter (KKF) and the Kernel Ridge Regression (KRR) in~\cite{romero2017kernel} are also reported in Fig.~\ref{fig_11}. We observe that the single-task diffusion LMSN and multitask diffusion LMS were not able to reconstruct the true temperature, whereas the multitask diffusion PLMS and the multitask diffusion LMSN showed a good reconstruction performance, and performed better than the KKF and KRR at the selected unobserved node. To evaluate the performance over all unobserved nodes on the test set, we considered the normalized mean square error (NMSE) defined as: 
\begin{equation}
    \textsc{NMSE}=\frac{\sum_{i=T_{train}+1}^{T} \| \diag(\One_\mathcal{\bar{S}}) \left(\by(i)-\hat{\by}(i)\right)  \|_2^2 }  {\sum_{i=T_{train}+1}^{T}\| \diag(\One_\mathcal{\bar{S}})  \by(i)   \|_2^2}
\end{equation}
where $\hat{\by}(i)$ denotes the reconstructed estimate at time $i$, $\One_\mathcal{\bar{S}}$ is the set indicator vector whose $k$-th entry is equal to one if node $k$ has not been sampled, and zero otherwise. The results are reported in Table~\ref{table:1}. We observe that the multitask diffusion PLMS performed as well as the LMSN at a lower computational cost, and both performed better than the KFF and KRR.
Finally, Figure~\ref{fig_12} reports the original topology and the clusters learned by the multitask diffusion PLMS.
\begin{figure}[!ht]
	\centering
	\includegraphics[scale=0.5]{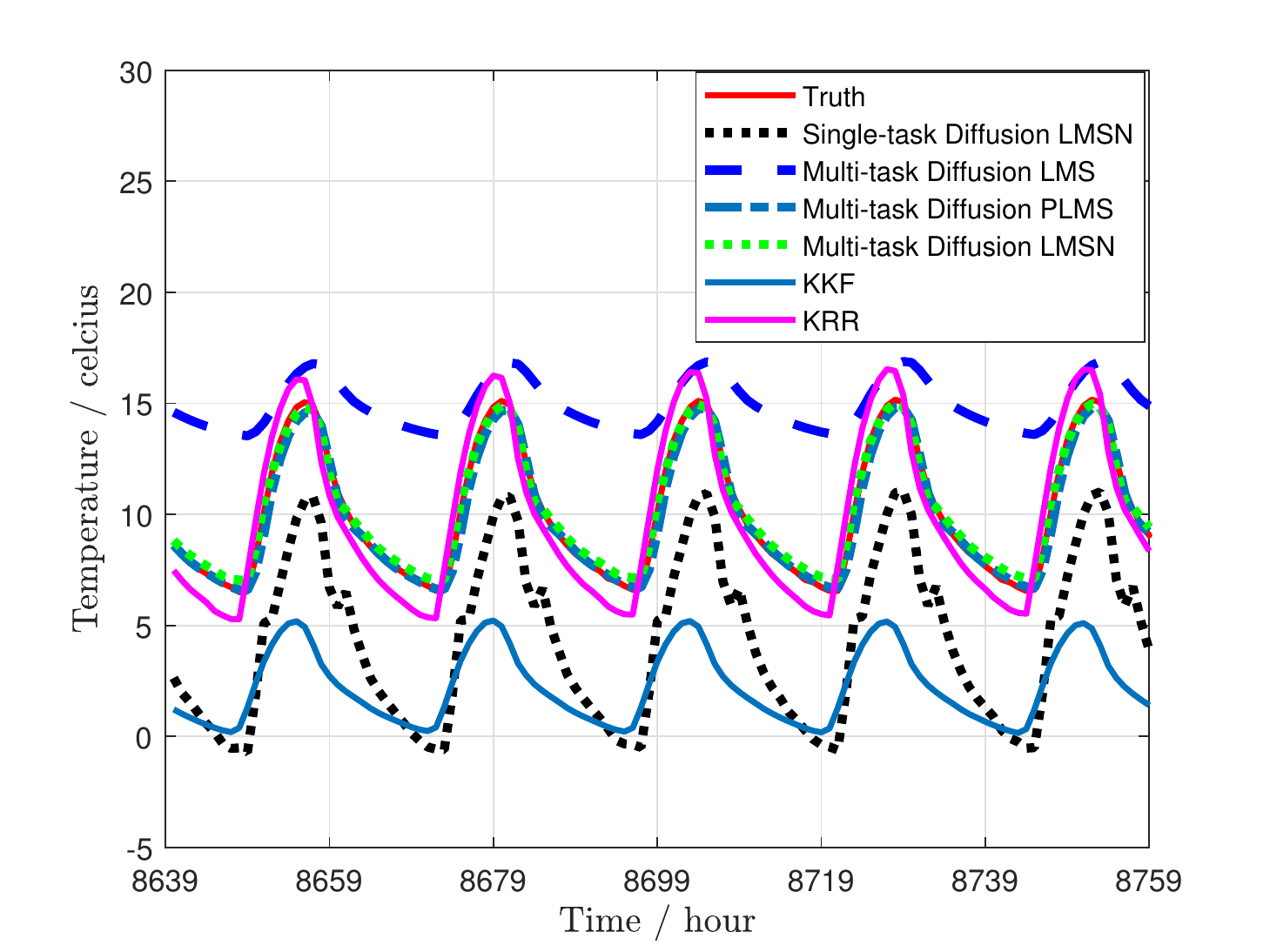}
	\caption{True temperatures and reconstructed ones at an unobserved node. $\mu_{\mathrm{LMS}}=10^{-5}$, $\mu_{\mathrm{PLMS}}=\mu_{\mathrm{LMSN}}=10^{-4}$.}
	\label{fig_11}
\end{figure} 

\begin{table}[!ht]
  \renewcommand{\arraystretch}{1.2}
  \caption{NMSE of different algorithms.}
  \label{table:1}
  \centering 
  \begin{tabular}{cc}
  \hline
  Algorithm & NMSE\\
  \hline
   KKF& $0.1093$\\
  KRR& $0.0479$\\
  Multitask diffusion LMS& $0.1152$\\
  Multitask diffusion PLMS& $0.0090$\\
  Multitask diffusion LMSN& $0.0031$\\
  \hline
  \end{tabular}  
  \end{table}

\begin{figure}[!ht]
	\centering
	\subfigure[Adjacency Matrix]{
		\begin{minipage}[t]{0.45\linewidth}
			\centering
			\includegraphics[scale=0.25]{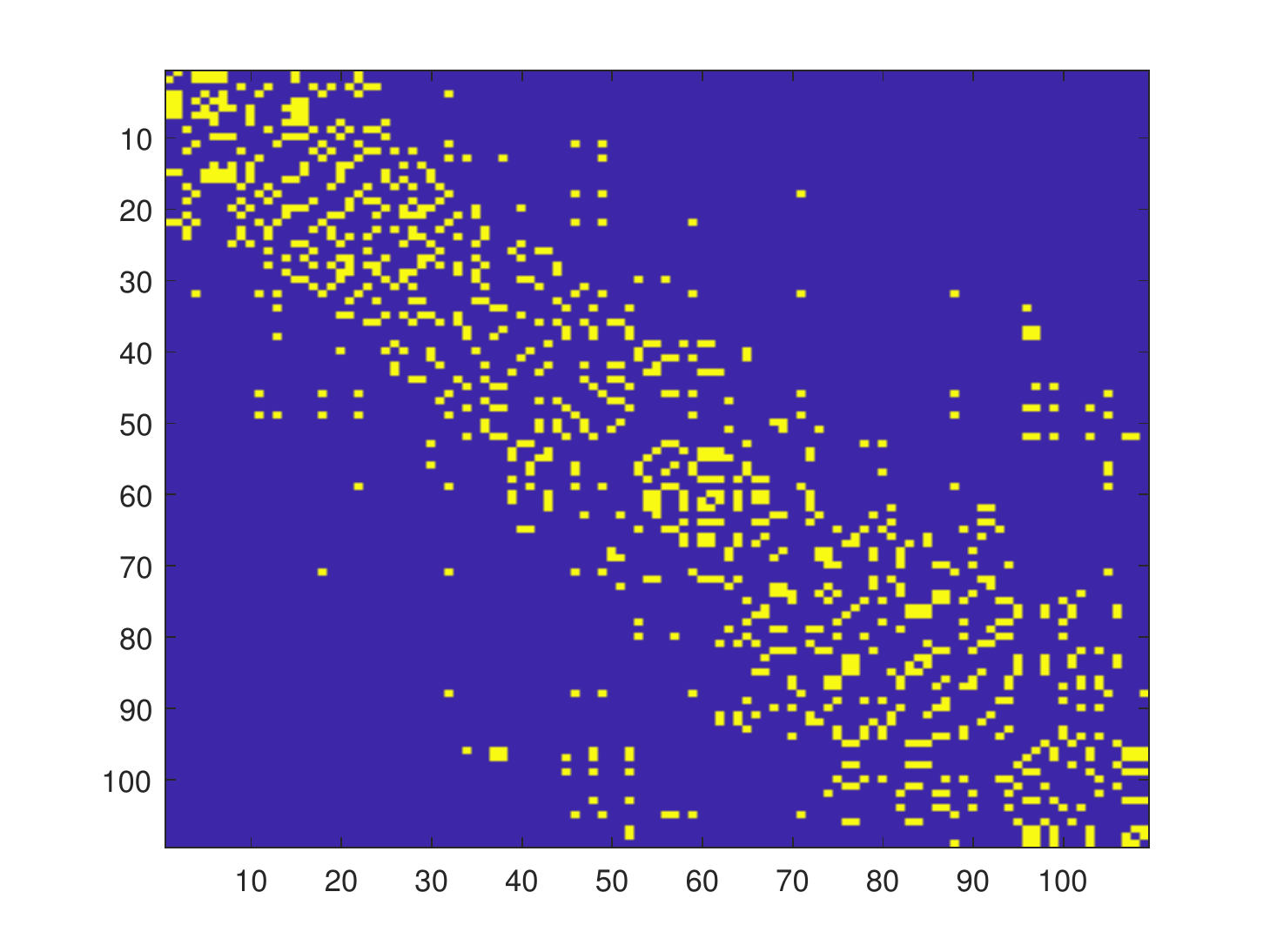}
		\end{minipage}%
	}%	
	\subfigure[Inferred clusters]{
		\begin{minipage}[t]{0.45\linewidth}
			\centering
			\includegraphics[scale=0.25]{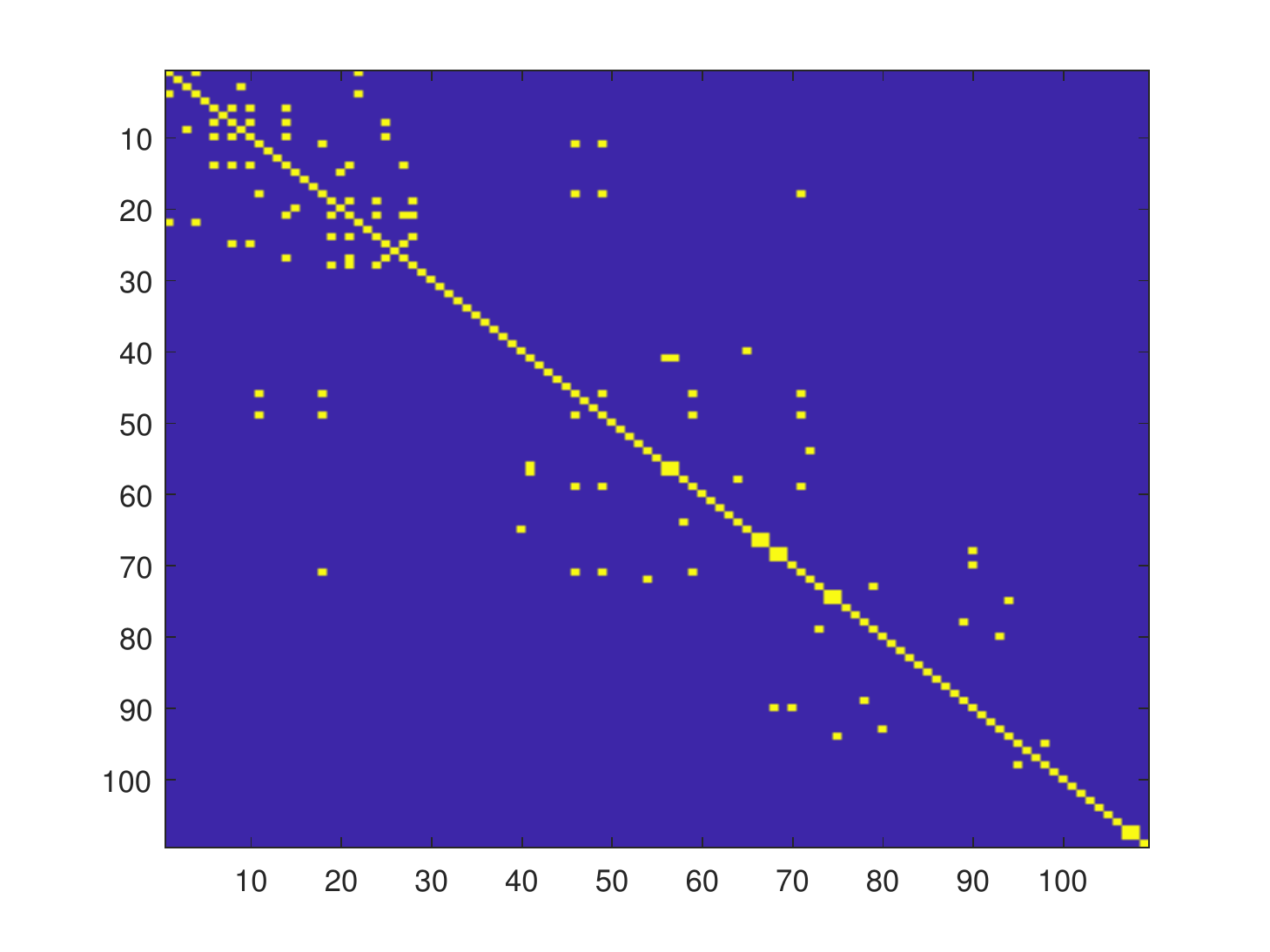}
		\end{minipage}%
	}%		
	\caption{ U.S. temperature graph topology and learned clusters. }
	\label{fig_12}
\end{figure}

In the second experiment, we divided the dataset into two parts. The first part contained the first 4200 hours sampled at the nodes showed in Fig.~\ref{fig_10} (a), and the second part contained the remaining hours sampled at the nodes showed in Fig.~\ref{fig_10} (b). This means that the sampling set abruptly changed at time $t=4201$ (the black circled node was unobserved in both case). We applied the multitask diffusion PLMS method over the entire dataset. In Fig.~\ref{fig_13} (a), the reconstructed temperature at the unobserved black circled node is reported from time $t=4100$ to $t=4300$ by using the filter coefficients learned up to time $t=4000$. As expected, we can notice that the reconstruction performance was successful from $t=4100$ to $t=4200$, and dramatically deteriorated after $t=4201$. This is due to the fact that the sampling set changed, which led to a drift in the filter coefficients to estimate. Figure~\ref{fig_13} (b) depicts the reconstruction behavior over last 120 hours using the filter coefficients learned up to time $t=8600$. It can be observed that  the proposed method was able to track the drift in the filter coefficients.
\begin{figure}[!ht]
	\centering
	\subfigure[]{
		\begin{minipage}[t]{0.45\linewidth}
			\centering
			\includegraphics[scale=0.3]{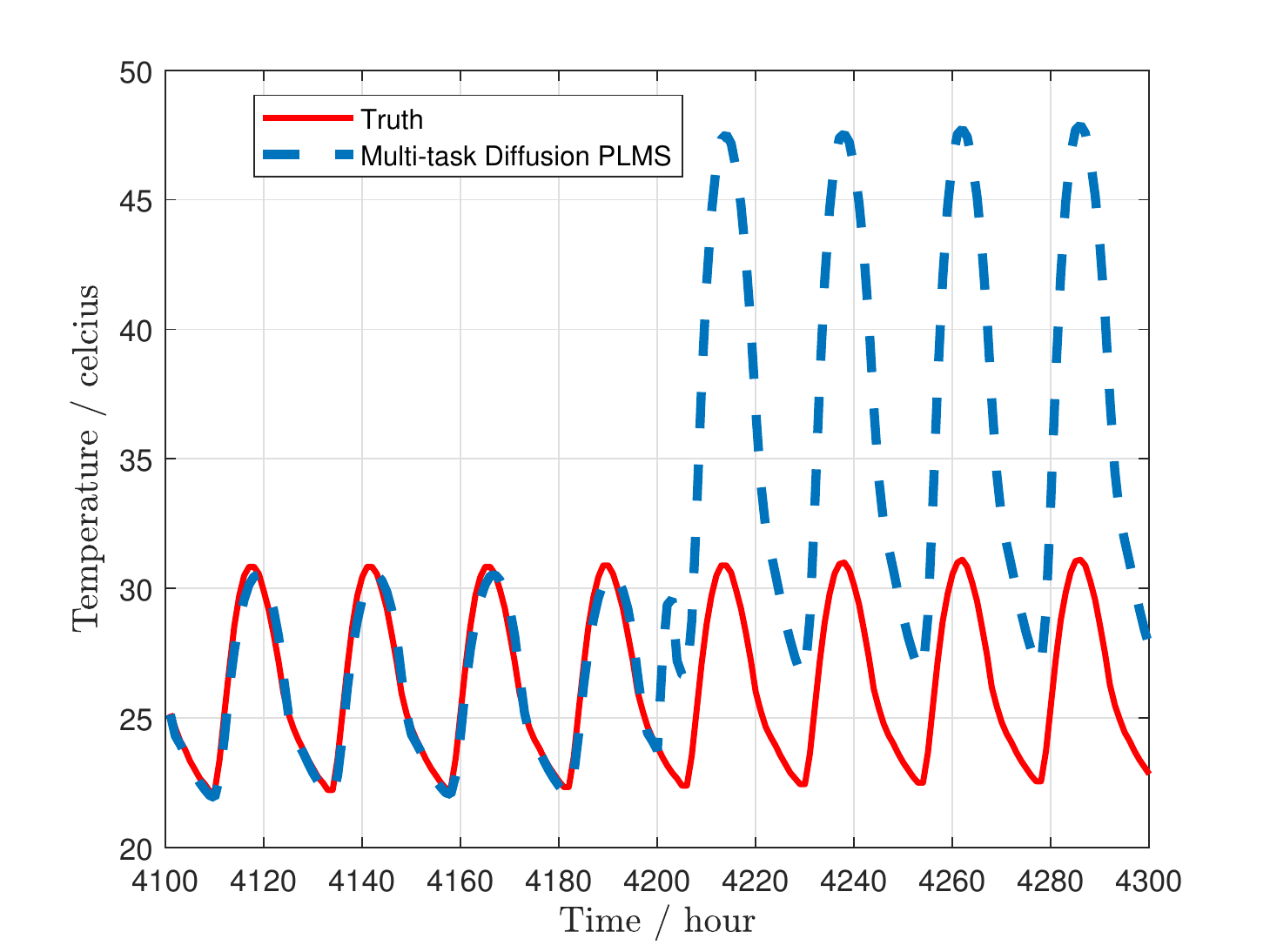}
		\end{minipage}%
	}%	
	\subfigure[]{
		\begin{minipage}[t]{0.45\linewidth}
			\centering
			\includegraphics[scale=0.3]{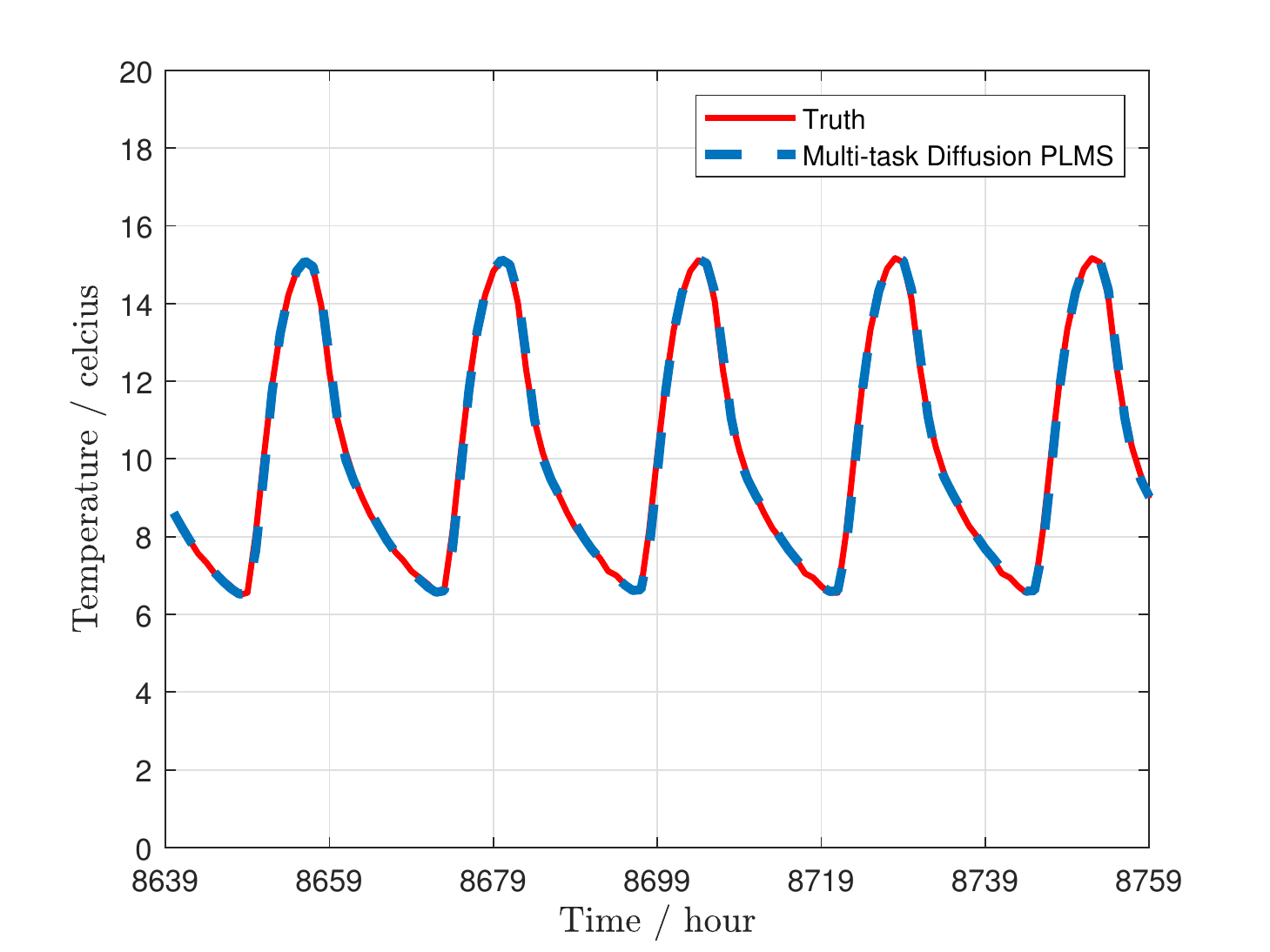}
		\end{minipage}%
	}%		
	\caption{True temperatures and reconstructed ones at an unobserved node. For clarity purposes, focus on the intervals (a) $[4100, 4300]$ and (b) $[8640, 8759]$.}
	\label{fig_13}
\end{figure}

%------------------	
\section{Conclusion}

Diffusion LMS strategies were considered to estimate graph filter coefficients in an adaptive and distributed manner. A diffusion LMS with Newton-like descent procedure was first proposed to achieve improved convergence rate, since usual algorithms may suffer from ill conditioning effects due to the use of non-energy preserving graph shift operators. A preconditioned diffusion LMS strategy, which does not require computationally intensive matrix inversion and only uses local information, was then devised to reduce the computational burden. Its convergence behavior was analyzed in the mean and mean-square-error sense. Finally, for hybrid node-varying graph filters, a clustering mechanism to be used with the preconditioned diffusion LMS was proposed. Simulation results validated the theoretical models and showed the efficiency of the proposed algorithms. In this work, we assumed that the graph shift operator is known and time invariant. In future works, we will consider ways to estimate the graph shift operator and the filter coefficients simultaneously. Because of their improved flexibility, which allows to address a variety of non-linear identification problems, we will also focus on non-parametric methods such as \cite{koppel_decentralized_2018,pradhan2018exact} and see how they can be incorporated into GSP framework.
	
	
	\bibliographystyle{IEEEtran}
	\bibliography{IEEEabrv,ref}
	
\end{document}